\documentclass{lailapreprint}

\title{Equilibria and Learning in Competitive Market Making}

\author{%
  \Name{Matteo Onger} \Email{matteo.onger@unimi.it}\\
  \Name{Luigi Foscari} \Email{luigi.foscari@unimi.it}\\
  \addr Department of Computer Science, University of Milan
}


\begin{document}

\maketitle

\begin{abstract}
    We propose a market-making game played among \(N \ge 2\) strategic market makers and a single non-strategic representative trader. In a continuous quote space, standard finite-game arguments do not apply, while non-convex and discontinuous payoffs rule out classical continuous existence results and complicate learning. Nonetheless, we characterize pure-strategy Nash equilibria (PSNEs) and show that at least one PSNE always exists under minimal assumptions on the underlying asset value distribution, and that the equilibrium bid--ask spread widens strictly with asset value volatility. To obtain a tractable learning framework, we discretize prices and, under an additional regularity condition, show that the time-averaged play of no-internal-regret learners approaches the continuous Nash price envelope as the tick size vanishes. Finally, we empirically validate our theoretical results through numerical simulations.
\end{abstract}




\section{Introduction} \label{sec:intro}

Classical models of price competition, such as the Bertrand oligopoly model~\citep{bertrand1883review}, elegantly explain why competition may drive prices toward marginal cost. However, applying this one-sided logic directly to financial market making overlooks a fundamental structural feature of liquidity provision: market makers operate simultaneously on both the buy and sell sides of the market, and their quotes actively shape the direction of incoming order flow.

This two-sided interaction is particularly prominent in financial markets characterized by \emph{inelastic trading demand}, such as commodity markets~\citep{bogmans2024power}, where overall trading volume does not respond to price changes and market participants simply choose whether to buy or sell. In conventional Bertrand competition, a firm's price reduction allows it to gain market share at the cost of a lower margin; in contrast, liquidity providers in inelastic markets post bid and ask quotes that establish a market \emph{mid-price}, which acts as a threshold determining whether buying or selling is advantageous relative to the asset's fundamental value, ultimately dictating trade direction and, indirectly, execution priority.

To study this mechanism, and how makers learn in it, we model liquidity provision as a static game played among \(N \ge 2\) strategic market makers and a single non-strategic representative trader. The former simultaneously post bid--ask quotes for a single asset with uncertain fundamental value, while the latter, upon observing the posted quotes, submits either a buy or a sell order according to a fixed decision rule that compares the realized fundamental value of the asset with the market mid-price induced by the quoted prices.

\paragraph{Contributions.}
By bridging Bertrand-style oligopoly~\citep{bertrand1883review} and classic market-making models \`{a} la~\citet{glosten1985bid}, our framework provides a tractable baseline for studying competition and learning dynamics in two-sided inelastic markets. To the best of our knowledge, this work provides the first theoretical and algorithmic analysis of market making in this setting. Our main contributions are summarized below:
\begin{itemize}
    \item \textbf{Nash Equilibrium Existence (\Cref{sec:equilibrium}):} We characterize pure-strategy Nash equilibria (PSNEs) in continuous quote spaces and prove that at least one PSNE always exists with minimal assumptions on the underlying asset value distribution. This equilibrium establishes a zero-profit competitive benchmark that guarantees perfect parity among all market makers, serving as an explicit reference model for competitive liquidity provision.
    
    \item \textbf{Impact of Volatility on Competitive Pricing (\Cref{subsec:asset_variance}):} We investigate how fundamental market uncertainty shapes equilibrium quotes. For location-scale families of asset distributions, we derive a structural scaling result: the equilibrium bid--ask spread scales linearly with asset value standard deviation. Crucially, this demonstrates that spread widenings can arise purely from competitive risk adjustments, rather than algorithmic collusion~\citep{calvano2020artificial, harrington2018developing} or a higher share of informed traders~\citep{glosten1985bid}.
    
    \item \textbf{Learning the Competitive Benchmark (\Cref{sec:learning}):} To obtain a tractable learning framework, we analyze learning over a discrete price grid with tick size \(\Delta\). Under an additional regularity condition, we prove that all grid correlated equilibria uniformly concentrate on the continuous Nash price envelope as \(\Delta \to 0\). Consequently, the time-averaged play of adaptive market makers using no-internal-regret learning algorithms is guaranteed to asymptotically approach the continuous competitive benchmark.
\end{itemize}
Finally, theoretical predictions and learning dynamics are empirically validated through numerical simulations in \Cref{sec:simulations}.


\subsection{Technical Challenges} \label{subsec:challenges}

Analyzing this market-making game presents distinct theoretical and algorithmic challenges:
\begin{itemize}
    \item \textbf{Two-Sided Coupling:} Updating a quote shifts the implied mid-price, directly coupling execution risk across both sides of the book. As a result, the game cannot be reduced to two independent Bertrand competitions. Indeed, since strategies consist of quote pairs, market makers can execute cross-side deviations (e.g., aggressively lowering the ask while retreating on the bid), thereby expanding the space of potentially profitable strategies in ways that break standard single-side equilibrium arguments.
    
    \item \textbf{Continuous vs.\ Discrete Tension:} Payoffs in this game are inherently non-convex and discontinuous. In continuous quote spaces, this rules out standard PSNE existence theorems and complicates learning dynamics. Discretizing prices resolves these hurdles, but it introduces grid-dependent equilibria that obscure the model's true structure and parameter dependencies.
\end{itemize}


\subsection{Related Work} \label{subsec:related_work}
 
This work connects to market microstructure, game theory, and learning in games; we review the most closely related literature.
 
\paragraph{Price Competition.}
Among classical oligopoly models~\citep{cournot1838recherches, bertrand1883review}, Bertrand competition is the closest benchmark to our setting: a marginal price improvement causes a discontinuous jump in demand, which complicates equilibrium existence. Pure-strategy equilibria can be recovered under additional assumptions, such as quasiconcave payoffs~\citep{dasgupta1986existence, reny1999existence} or suitable cost and demand conditions, possibly with multiple equilibria~\citep{dastidar1995existence}. Neither applies here: in particular, quasiconcavity fails for bid--ask pairs, since averaging two profitable one-sided deviations can leave a maker inactive on both sides. Our existence proof instead exploits the monotonicity of the midpoint operator via Tarski's theorem~\citep{tarski1955lattice}, as in games with strategic complementarities, where extremal equilibria likewise bracket the equilibrium set~\citep{milgrom1990rationalizability, vives1990nash}, although our game is not supermodular.
 
\paragraph{Market Microstructure.}
Seminal models study how informed order flow~\citep{glosten1985bid, kyle1985continuous, easley1987price} and inventory risk~\citep{ho1981optimal} shape spreads, typically assuming that dealers break even. Our zero-profit quotes mirror the upper-tail pricing of \citet{glosten1994electronic}, but arise as a Nash outcome among finitely many strategic makers, as in models of competing dealers~\citep{dennert1993price, biais2000competing, back2013strategic}. Unlike these works, which rely on noise traders to avoid no-trade outcomes~\citep{milgrom1982information}, we consider a fully informed trader with inelastic demand, so that quotes determine the direction of trade rather than its occurrence. Spreads are also well known to widen with volatility~\citep{copeland1983information, ho1981optimal}; we show that, within location-scale families, they scale exactly linearly. Finally, our grid game relates to the tick-size literature~\citep{anshuman1998market, cartea2022algorithms}, whose small-tick limit we analyze.
 
\paragraph{Learning in Games and Market Making.}
Our guarantees build on online learning~\citep{cesa2006prediction, orabona2023modern} and on the convergence of no-internal-regret dynamics to correlated equilibria~\citep{foster1997calibrated, hart2000simple, blum2007external}. Unlike \citet{stoltz2007learning}, who learn over compact strategy sets with continuous payoffs, we approximate our discontinuous game by grid games. In non-strategic environments, online learning has been used to learn market-making strategies~\citep{abernethy2013adaptive, cesa2025market, maran2026online}, via reductions to first-price auctions~\citep{achddou2021fast, cesa2024role, han2024optimal} and dynamic pricing~\citep{cesa2019dynamic, tullii2025improved}, while agent-based studies train makers via reinforcement learning~\citep{spooner2018market, ganesh2019reinforcement}; here, we characterize the equilibria that interacting learners approach.
 
\paragraph{Algorithmic Collusion.}
Learning agents can sustain supra-com\-petitive prices~\citep{calvano2020artificial, klein2021autonomous, banchio2022artificial, hansen2021algorithmic}, also in dealer markets~\citep{cont2024dynamics, colliard2026algorithmic}, raising regulatory concerns~\citep{harrington2018developing}; no-swap-regret learning rules out some of these outcomes~\citep{arunachaleswaran2025algorithmic}, motivating our choice of learners. \citet{hartline2024regulation} detect collusion from prices alone, a test later shown to be confoundable~\citep{hartline2025regulation}; in this vein, we identify one such confound: volatility-driven spread widening.


\section{Model} \label{sec:model}

In this section, we present the model of inelastic markets used throughout the paper. Let \(N \ge 2\) and \(\mathcal{N} \coloneq \{1, \dots, N\}\) denote the set of market makers. The fundamental value of the asset is modeled as a continuous, non-degenerate random variable \(V\), governed by an unknown cumulative distribution function \(F\) (with density \(f\) where it exists) and bounded support \([\underline{v}, \overline{v}] \subseteq [0,1]\). Crucially, market makers possess no prior knowledge of the distribution \(F\) and must learn its structure over repeated interactions. Conversely, the trader observes the realized value \(v\) of \(V\) prior to placing an order.

At the start of each round, each market maker \(i \in \mathcal{N}\) posts a pair \(s_i \coloneq (b_i, a_i)\), consisting of a bid quote \(b_i\) and an ask quote \(a_i\), with \(0 \le b_i \le a_i \le 1\). The strategy space for market maker \(i\) is thus
\[
    \mathcal{S}_i
    \coloneq
    \left\{(b,a) \in [0,1]^2 \;\middle|\; 0 \le b \le a \le 1 \right\}.
\]
Let \(\mathbf{s} \coloneq (s_1, \dots, s_N) \in \mathcal{S} \coloneq \prod_{i \in \mathcal{N}} \mathcal{S}_i\) denote the joint strategy profile, and let \(\mathbf{b} \coloneq (b_1, \dots, b_N)\) and \(\mathbf{a} \coloneq (a_1, \dots, a_N)\) denote the joint vectors of bid and ask quotes, respectively, so that \(\mathbf{s} = (\mathbf{b}, \mathbf{a})\). We write \(\mathbf{s}_{-i}\) for the profile of all makers other than \(i\). The prevailing market prices are determined by the best available quotes:
\[
    a^*(\mathbf{a}) \coloneq \min_{i \in \mathcal{N}} a_i
    \qquad \text{and} \qquad
    b^*(\mathbf{b}) \coloneq \max_{i \in \mathcal{N}} b_i,
\]
which together induce the market mid-price
\[
    m^*(\mathbf{s}) \coloneq \dfrac{a^*(\mathbf{a}) + b^*(\mathbf{b})}{2}.
\]

After observing the posted quotes, the trader chooses an action \(d \in \mathcal{D} \coloneq \{-1, +1\}\), where \(d = +1\) denotes a buy order and \(d = -1\) denotes a sell order. Specifically, upon observing the asset's realized fundamental value \(v\) and comparing it with the market mid-price \(m^*(\mathbf{s})\), the trader selects an action \(d \in \mathcal{D}\) according to the following decision rule:
\begin{equation} \label{eq:trader_rule}
    \mathbb{P}(d = +1 \mid v, \mathbf{s}) \coloneq
    \begin{cases}
        1 & \text{if } v > m^*(\mathbf{s}), \\
        0 & \text{if } v < m^*(\mathbf{s}), \\
        \frac{1}{2} & \text{if } v = m^*(\mathbf{s}),
    \end{cases}
\end{equation}
with \(\mathbb{P}(d = -1 \mid v, \mathbf{s}) = 1 - \mathbb{P}(d = +1 \mid v, \mathbf{s})\).

To complete the model specification, we define order allocation and execution payoffs. Let \(n_A(\mathbf{a}) \coloneq |\{i \in \mathcal{N} \mid a_i = a^*(\mathbf{a})\}|\) and \(n_B(\mathbf{b}) \coloneq |\{i \in \mathcal{N} \mid b_i = b^*(\mathbf{b})\}|\) denote the number of market makers quoting the best ask and bid prices, respectively. If multiple market makers quote the same prevailing price, the incoming order is allocated equally among them.\footnote{We adopt equal-split allocation as it is equivalent in expectation to selecting a tied market maker uniformly at random, while eliminating tie-breaking stochasticity, thereby reducing reward variance in simulation and learning.} Consequently, for a given strategy profile \(\mathbf{s}\), the expected payoff of market maker \(i \in \mathcal{N}\) is \(u_i(\mathbf{s}) \coloneq \mathbb{E}_{V, d \sim \mathbb{P}(\cdot \mid V, \mathbf{s})} [r_i(\mathbf{s}; V, d)]\), where the ex-post payoff \(r_i\) under realization \(v\) and trader order \(d\) is given by
\begin{equation} \label{eq:payoff_func}
    r_i(\mathbf{s}; v, d) \coloneq
    \begin{cases}
        \dfrac{a_i - v}{n_A(\mathbf{a})} & \text{if } d = +1 \text{ and } a_i = a^*(\mathbf{a}), \\[2ex]
        \dfrac{v - b_i}{n_B(\mathbf{b})} & \text{if } d = -1 \text{ and } b_i = b^*(\mathbf{b}), \\[2ex]
        0 & \text{otherwise.}
    \end{cases}
\end{equation}
For completeness, the ex-post realized payoff of the trader is equal to \((v - a^*(\mathbf{a})) \, \mathbb{I}_{\{d=+1\}} + (b^*(\mathbf{b}) - v) \, \mathbb{I}_{\{d=-1\}}\), implying that the overall trade interaction between the aggregate pool of market makers and the trader is zero-sum.

\paragraph{Strategic Scope and Game Reduction.}
Given the optimal decision rule in~\eqref{eq:trader_rule}, the trader's action is fully determined by the realized fundamental value and the posted quotes, leaving liquidity providers as the sole strategic players. Isolating their interactions yields a reduced general-sum game among the \(N\) market makers, where a buy order occurs exactly when \(V > m^*(\mathbf{s})\), and a sell order exactly when \(V < m^*(\mathbf{s})\).\footnote{Since \(V\) is a continuous random variable, the event \(V = m^*(\mathbf{s})\) occurs with probability zero, so this holds almost surely.} Each maker's expected payoff accordingly decomposes as \(u_i(\mathbf{s}) = u_i^{\mathrm{ask}}(\mathbf{s}) + u_i^{\mathrm{bid}}(\mathbf{s})\), where
\[
    u_i^{\mathrm{ask}}(\mathbf{s})
    \coloneq
    \dfrac{\mathbb{I}_{\{a_i = a^*(\mathbf{a})\}}}{n_A(\mathbf{a})} \,
    \mathbb{P}\left( V > m^*(\mathbf{s}) \right)
    \left( a^*(\mathbf{a}) - \mathbb{E}\left[ V \;\middle|\; V > m^*(\mathbf{s}) \right] \right)
\]
is maker \(i\)'s expected payoff from buy orders, and \(u_i^{\mathrm{bid}}(\mathbf{s})\) is defined symmetrically for sell orders. This reduced \(N\)-agent formulation fully captures the competitive dynamics of the market, providing the exact setting analyzed in the following sections.


\section{Equilibrium Concept} \label{sec:equilibrium}

In this section, we analyze the equilibrium properties of the proposed market-making game. Our main goal is to prove that at least one pure-strategy Nash equilibrium always exists. Establishing existence is nontrivial in our setting, as the strategy space is continuous and the payoff functions are non-convex and discontinuous due to Bertrand-style price competition and tie-breaking in order allocation, precluding a direct application of standard equilibrium existence results.

\begin{definition}[Pure-Strategy Nash Equilibrium] \label{def:psne}
    A strategy profile \(\mathbf{s}\) is a \emph{pure-strategy Nash equilibrium} (PSNE) if, for every market maker \(i \in \mathcal{N}\) and every feasible deviation \(\widetilde{s}_i \coloneq (\widetilde{b}_i, \widetilde{a}_i) \in \mathcal{S}_i\), we have:
    \[
        u_i(\widetilde{s}_i, \mathbf{s}_{-i}) \le u_i(s_i, \mathbf{s}_{-i}).
    \]
\end{definition}

In our two-sided market setting, a deviation \(\widetilde{s}_i\) allows market maker \(i\) to unilaterally alter its ask quote, its bid quote or both, simultaneously. As previously mentioned, such a deviation can affect not only the execution price and the maker's share of the existing order flow, but also the composition of that flow by shifting the induced mid-price and hence the trader's buy--sell decision. This coupling between the two sides of the market enlarges the set of potentially profitable deviations and makes their analysis more complex than in standard Bertrand competition. Establishing that a profile is a PSNE therefore also requires ruling out deviations that may exploit this interdependence.

To establish existence, we first derive the necessary conditions that a Nash equilibrium profile must satisfy.

\begin{observation} \label{obs:zero_profit}
    In any PSNE, the expected profit from each side of the book must be zero independently. To see this, first observe that a strategy yielding negative expected payoff cannot be optimal, since quoting the boundary prices, i.e., \(s_i = (0, 1)\), guarantees a non-negative payoff for every realization of \(V\). Hence, the expected profit must be non-negative. On the other hand, strictly positive expected profits cannot be sustained: if a side of the book generated strictly positive expected profit, a competing maker could undercut the corresponding quote by an arbitrarily small amount and capture all of the resulting order flow at a slightly less favorable price. Finally, because makers can adjust their quotes on the two sides independently, positive profits on one side cannot compensate for losses on the other: a maker could withdraw from the loss-making side while retaining the profitable side.
\end{observation}

\begin{observation} \label{obs:best_price_competition}
    As in standard Bertrand competition, any PSNE must have at least two market makers quoting the best ask and bid, i.e., \(n_A(\mathbf{a}), n_B(\mathbf{b}) \ge 2\). The usual undercutting argument must account for the fact that quotes also affect the mid-price. If, for instance, \(n_A(\mathbf{a}) = 1\), the sole ask maker can execute a joint deviation by raising its ask by a small \(\varepsilon > 0\), while improving its bid by a much smaller amount. By shifting the mid-price upward, this deviation increases the margin on buy orders and converts marginal buy orders into sell orders captured by its new bid, yielding a strictly positive expected profit.
\end{observation}

Building on the observations above, consider a PSNE profile \(\mathbf{s}\), and let \(a^*\), \(b^*\) and \(m^* \coloneq (a^* + b^*)/2\) denote its best ask, best bid and induced mid-price, respectively. The \emph{zero-profit condition} implies
\[
    \mathbb{E}\left[a^* - V \;\middle|\; V > m^*\right] = 0,
    \qquad
    \mathbb{E}\left[V - b^* \;\middle|\; V < m^*\right] = 0,
\]
and hence
\[
    a^* = \mathbb{E}\left[V \;\middle|\; V > m^*\right],
    \qquad
    b^* = \mathbb{E}\left[V \;\middle|\; V < m^*\right].
\]
Thus, at any PSNE, the zero-profit condition uniquely determines the best quotes for its mid-price \(m^*\). More generally, for any candidate equilibrium mid-price \(\widehat{m}\), internal consistency requires that the midpoint induced by the corresponding \emph{zero-profit quotes} coincide with \(\widehat{m}\) itself, motivating the following definition.

\begin{definition}[Midpoint Operator] \label{def:midpoint_operator}
    For any \(m \in [\underline{v}, \overline{v}]\), define the zero-profit quotes
    \[
        a(m) \coloneq \mathbb{E}\left[V \;\middle|\; V > m\right],
        \qquad
        b(m) \coloneq \mathbb{E}\left[V \;\middle|\; V < m\right]
    \]
    whenever the conditioning events have positive probability, and set \(a(\overline{v}) \coloneq \overline{v}\) and \(b(\underline{v}) \coloneq \underline{v}\) otherwise, in accordance with their one-sided limits. We then define the \emph{midpoint operator} \(H : [\underline{v}, \overline{v}] \to [\underline{v}, \overline{v}]\) as
    \[
        H(m) \coloneq \dfrac{a(m) + b(m)}{2}.
    \]
\end{definition}

By construction, \(H(m)\) is the midpoint induced by the zero-profit quotes associated with \(m\). Consequently, the mid-price \(m^*\) of any PSNE must be self-consistent, satisfying
\[
    m^* = H(m^*),
\]
i.e., any PSNE mid-price is necessarily a \emph{fixed point} of \(H\).

We denote by \(\mathcal{F} \coloneq \{ m \in [\underline{v}, \overline{v}] \mid m = H(m) \}\) the set of fixed points of \(H\). Next, we show that \(\mathcal{F}\) is nonempty and lies strictly between the lower and upper bounds of the support of \(V\).

\begin{theorem}[Existence of a Zero-Profit Midpoint] \label{thm:existence_zero_profit_midpoint}
    The set of fixed points \(\mathcal{F}\) of \(H\) is nonempty and satisfies \(\mathcal{F} \subset (\underline{v},\overline{v})\).
\end{theorem}
\begin{proof}
    We first establish that \(H\) is an order-preserving self-map on the complete lattice \([\underline{v},\overline{v}]\). For every \(m \in [\underline{v},\overline{v}]\), \(a(m)\) and \(b(m)\) are well defined. Moreover, both functions are nondecreasing in \(m\): increasing \(m\) removes lower values from the conditioning event defining \(a(m)\), while adding higher values to the conditioning event defining \(b(m)\). Consequently, \(H(m) = (a(m)+b(m))/2\) is also nondecreasing, establishing that \(H\) is indeed an order-preserving self-map on \([\underline{v},\overline{v}]\), as claimed.
    
    By the Knaster--Tarski fixed-point theorem~\citep{tarski1955lattice}, \(H\) admits at least one fixed point \(\widehat{m} \in [\underline{v},\overline{v}]\) satisfying \(\widehat{m} = H(\widehat{m})\), so \(\mathcal{F} \neq \varnothing\).

    It remains to rule out boundary fixed points. At the upper boundary, we have
    \[
        H(\overline{v}) = \frac{\overline{v} + b(\overline{v})}{2} < \overline{v},
    \]
    since \(b(\overline{v}) = \mathbb{E}[V] < \overline{v}\). Symmetrically, at the lower boundary,
    \[
        H(\underline{v}) = \frac{a(\underline{v}) + \underline{v}}{2} > \underline{v},
    \]
    since \(a(\underline{v}) = \mathbb{E}[V] > \underline{v}\). Consequently, neither boundary point can be a fixed point of \(H\), so \(\mathcal{F} \subset (\underline{v},\overline{v})\).
\end{proof}

Although the zero-profit condition is necessary for a PSNE, it is not sufficient. For any \(\widehat{m} \in \mathcal{F}\), we denote by \(s(\widehat{m}) \coloneq (b(\widehat{m}), a(\widehat{m}))\) the \emph{zero-profit strategy} at \(\widehat{m}\), and by \(\mathbf{s}(\widehat{m}) \coloneq (s(\widehat{m}), \dots, s(\widehat{m}))\) the symmetric \emph{zero-profit profile} where all market makers play \(s(\widehat{m})\). Even when all competitors stick to \(\mathbf{s}_{-i}(\widehat{m})\), a maker \(i\) may still find a profitable unilateral deviation \(\widetilde{s}_i \coloneq (\widetilde{b}_i, \widetilde{a}_i)\) that exploits the two-sided coupling of the market.

\begin{figure}[t]
    \centering
    \begin{tikzpicture}[
    x=0.75\linewidth, y=1cm,
    font=\scriptsize\sffamily,
    >={Stealth[length=5pt,width=4pt]}
]

    \def\vmin{0}
    \def\vmax{1}
    \def\bstar{0.25}
    \def\mstar{0.50}
    \def\astar{0.75}
    \def\bprime{0.12} 
    \def\aprime{0.60} 
    \def\y{0.425}     
    
    \fill[orange!20] (\y, -0.05) rectangle (\mstar, 0.05);
    
    \draw[thick, black] (\vmin-0.05, 0) -- (\vmax+0.05, 0) node[right] {$V$};
    
    \draw[thick] (\vmin, 0.1) -- (\vmin, -0.1) node[below] {$\underline{v}$};
    \draw[thick] (\vmax, 0.1) -- (\vmax, -0.1) node[below] {$\overline{v}$};
    
    \node[eq, anchor=south west, font=\bfseries\footnotesize] at (\vmin-0.05, 2.3) {Zero-Profit Profile};
    
    \draw[eq, dashed] (\bstar, 0) -- (\bstar, 0) node[above] {$b(\widehat{m})$};
    \draw[eq, dashed] (\mstar, 0) -- (\mstar, 1.8) node[above] {$\widehat{m}$};
    \draw[eq, dashed] (\astar, 0) -- (\astar, 0) node[above] {$a(\widehat{m})$};
    
    \fill[eq] (\bstar, 0) circle (1.5pt);
    \fill[eq] (\mstar, 0) circle (1.5pt);
    \fill[eq] (\astar, 0) circle (1.5pt);
    
    \draw[decorate, decoration={brace, amplitude=5pt, mirror}, eq, thick] 
        (\mstar-0.01, 1.4) -- (\vmin+0.01, 1.4) 
        node[midway, above=2pt] {$u^\mathrm{bid}_i(\mathbf{s}(\widehat{m})) = 0$};
        
    \draw[decorate, decoration={brace, amplitude=5pt, mirror}, eq, thick] 
        (\vmax-0.01, 1.4) -- (\mstar+0.01, 1.4) 
        node[midway, above=2pt] {$u^\mathrm{ask}_i(\mathbf{s}(\widehat{m})) = 0$};
    
    \draw[<-, eq, thick] (\vmin+0.05, 0.8) -- (\mstar-0.05, 0.8) node[midway, fill=white, inner sep=2pt, font=\scriptsize] {Sell Flow};
    \draw[->, eq, thick] (\mstar+0.05, 0.8) -- (\vmax-0.05, 0.8) node[midway, fill=white, inner sep=2pt, font=\scriptsize] {Buy Flow};
    
    \node[dev, anchor=north west, font=\bfseries\footnotesize] at (\vmin-0.05, -2.5) {Unilateral Deviation ($\widetilde{s}_i=(\widetilde{b}_i, \widetilde{a}_i)$)};
    
    \draw[dev, dashed] (\y, 0) -- (\y, -2) node[below] {$y$};
    \draw[dev, dashed] (\aprime, 0) -- (\aprime, 0) node[below] {$\widetilde{a}_i$};
    \draw[black!50, dashed] (\bprime, 0) -- (\bprime, 0) node[below, text=black!50] {$\widetilde{b}_i$};
    
    \fill[dev] (\y, 0) circle (1.5pt);
    \fill[dev] (\aprime, 0) circle (1.5pt);
    \fill[black!50] (\bprime, 0) circle (1.5pt);
    
    \draw[->, thick, dev, bend left=20] (\astar, -0.1) to node[below, font=\bfseries\tiny] {Margin Loss} (\aprime+0.02, -0.1);
    
    \draw[->, thick, orange!90!black, bend left=20] (\mstar, -0.1) to node[below, font=\bfseries\tiny] {Volume Gain} (\y, -0.1);
    
    \draw[decorate, decoration={brace, amplitude=5pt}, dev, thick] 
        (\vmax-0.01, -1.6) -- (\y+0.01, -1.6) 
        node[midway, below=2pt, align=center] {$u^\mathrm{ask}_i(\widetilde{s}_i, \mathbf{s}_{-i}(\widehat{m})) \stackrel{?}{>} 0$};
    
    \draw[decorate, decoration={brace, amplitude=5pt}, black!50, thick] 
        (\y-0.01, -1.6) -- (\vmin+0.01, -1.6) 
        node[midway, below=2pt, align=center, text=black!50] {$u^\mathrm{bid}_i(\widetilde{s}_i, \mathbf{s}_{-i}(\widehat{m})) = 0$};
    
    \draw[->, dev, thick] (\y+0.05, -1) -- (\vmax-0.05, -1) node[midway, fill=white, inner sep=2pt, text=dev, font=\scriptsize] {Captured Buy Flow};
    \draw[<-, black!50, thick] (\vmin+0.05, -1) -- (\y-0.05, -1) node[midway, fill=white, inner sep=2pt, text=black!50, font=\scriptsize] {Sell Flow (Competitors)};
\end{tikzpicture}
    \Description{
        Comparison between the baseline zero-profit state and a unilateral ask undercut deviation with bid withdrawal. The top panel shows the baseline profile with zero expected profits on both sides and induced midpoint \(\widehat{m}\). The bottom panel illustrates an undercut to \(\widetilde{a}_i\) paired with bid withdrawal \(\widetilde{b}_i\). This shifts the induced midpoint to \(y < \widehat{m}\), expanding the captured buy flow to \((y, \overline{v}]\) at the cost of a reduced per-trade margin.
    }
    \caption{
        Mechanism of a unilateral ask undercut paired with bid withdrawal. Lowering the ask to \(\widetilde{a}_i\) narrows the per-trade execution margin but shifts the induced midpoint to \(y < \widehat{m}\), thereby expanding the captured buy order flow across \((y, \overline{v}]\).
    }
    \label{fig:undercut_deviation}
\end{figure}
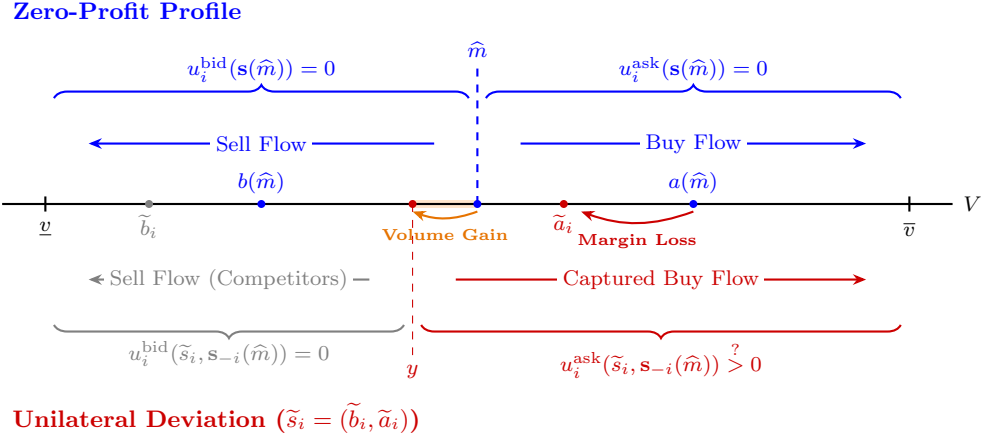

By improving a quote on one side, e.g., undercutting the ask to \(\widetilde{a}_i < a(\widehat{m})\) paired with bid withdrawal (\(\widetilde{b}_i < b(\widehat{m})\)), as illustrated in \Cref{fig:undercut_deviation}, maker \(i\) becomes the unique recipient of buy orders. This deviation shifts the induced midpoint to \(y \coloneq (\widetilde{a}_i + b(\widehat{m}))/2 < \widehat{m}\), yielding an expected ask payoff of
\[
    u_i^{\mathrm{ask}}(\widetilde{s}_i, \mathbf{s}_{-i}(\widehat{m}))
    = \mathbb{P}(V > y) \left( \widetilde{a}_i - a(y) \right).
\]
This reveals the trade-off of two-sided competition: undercutting the ask narrows the per-trade margin, but the leftward shift of the midpoint (from \(\widehat{m}\) to \(y\)) extends execution to buyers with lower values, raising the execution probability \(\mathbb{P}(V > y)\). A deviation is strictly profitable whenever this volume gain outweighs the lower margin. Crucially, this shift in the mid-price also impacts the opposite side of the book. Under the newly induced midpoint \(y < \widehat{m}\), maintaining the original bid quote \(b(\widehat{m})\), i.e., playing \((b(\widehat{m}), \widetilde{a}_i)\), would subject maker \(i\) to an expected bid payoff of
\[
    u_i^{\mathrm{bid}}((b(\widehat{m}), \widetilde{a}_i), \mathbf{s}_{-i}(\widehat{m}))
    = \frac{1}{N} \, \mathbb{P}(V < y) \left( b(y) - b(\widehat{m}) \right) \le 0,
\]
since \(b(y) \le b(\widehat{m})\). This exposes the maker to adverse selection on the bid side, making it strictly loss-making whenever \(b(y) < b(\widehat{m})\). To prevent these potential losses from reducing its ask profits, the deviating maker plays \(\widetilde{s}_i\) instead, actively withdrawing from the bid side, which ensures
\[
    u_i^{\mathrm{bid}}(\widetilde{s}_i, \mathbf{s}_{-i}(\widehat{m})) = 0.
\]

The next step is therefore to establish whether at least one point of \(\mathcal{F}\) is robust against such deviations. The following lemmas establish the properties needed to answer this question (their proofs are provided in \Cref{app:proofs_equilibrium}); combining them, we then identify such a point, which constitutes a PSNE mid-price.

\begin{restatable}{lemma}{singleSidedDeviations} \label{lem:single_sided_deviations}
    Let \(\widehat{m} \in \mathcal{F}\). If the associated zero-profit profile \(\mathbf{s}(\widehat{m})\) admits a profitable unilateral deviation, then it admits either a profitable ask undercut with bid withdrawal or a profitable bid raise with ask withdrawal.
\end{restatable}

The preceding lemma shows that, without loss of generality, it suffices to analyze these single-sided deviations. Hereafter, an ask undercut paired with bid withdrawal is referred to simply as an ask undercut (and symmetrically for bid raises).

\begin{restatable}{lemma}{implicationsDeviations} \label{lem:implications_deviations}
    Let \(\widehat{m} \in \mathcal{F}\). Then:
    \begin{enumerate}[label=(\roman*), beginpenalty=10000]
        \item For any \(x > \widehat{m}\), the unique bid raise inducing mid-price \(x\) is profitable if and only if \(2x - b(x) < a(\widehat{m})\), which implies \(H(x) > x\).
        
        \item For any \(y < \widehat{m}\), the unique ask undercut inducing mid-price \(y\) is profitable if and only if \(2y - a(y) > b(\widehat{m})\), which implies \(H(y) < y\).
    \end{enumerate}
\end{restatable}

\begin{restatable}{lemma}{extremalImmunity} \label{lem:extremal_immunity}
    Let \(m_{\min} \coloneq \min \mathcal{F}\) and \(m_{\max} \coloneq \max \mathcal{F}\). Then (i) no profitable ask undercut exists at \(m_{\min}\), and (ii) no profitable bid raise exists at \(m_{\max}\).
\end{restatable}

\begin{restatable}{lemma}{deviationsContagion} \label{lem:deviations_contagion}
    Let \(m_1, m_2 \in \mathcal{F}\) with \(m_1 < m_2\). Then:
    \begin{enumerate}[label=(\roman*), beginpenalty=10000]
        \item If a profitable bid raise at \(m_1\) induces a mid-price \(x > m_2\), then the bid raise inducing the same midpoint \(x\) at \(m_2\) is also profitable.
        
        \item If a profitable ask undercut at \(m_2\) induces a mid-price \(y < m_1\), then the ask undercut inducing the same midpoint \(y\) at \(m_1\) is also profitable.
    \end{enumerate}
\end{restatable}

As anticipated above, we now combine these properties to identify a fixed point that admits no profitable unilateral deviation and hence constitutes a PSNE mid-price.

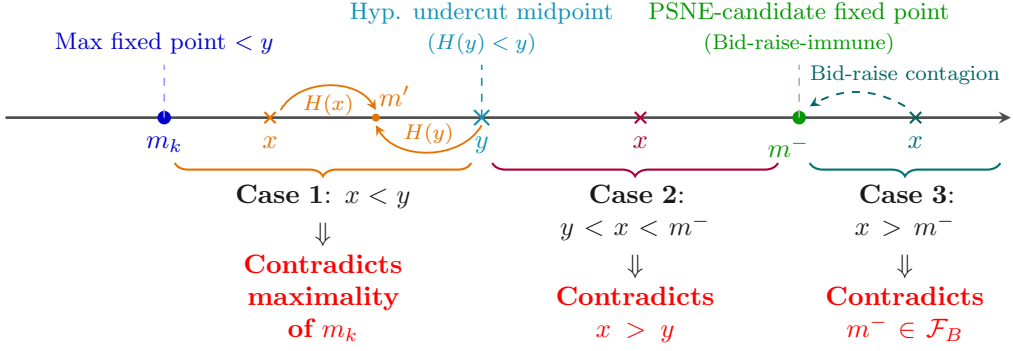
\begin{figure}[t]
    \centering
    \begin{tikzpicture}[thick, >=stealth, x=1.4cm]
    \tikzset{
        case_text/.style={
            align=center, 
            font=\small, 
            text=black!90, 
            anchor=north,
            inner sep=4pt
        },
        fixed_point/.style={
            circle, 
            fill=black, 
            inner sep=1.8pt
        },
        brace_style/.style={
            decorate, 
            decoration={brace, amplitude=5pt, mirror}, 
            line width=0.8pt
        }
    }
    
    \draw[->, line width=1pt, black!70] (0,0) -- (9.5,0);
    
    \node[fixed_point, fill=blue!80!black] (mk) at (1.5,0) {};
    \node[below=3pt, blue!80!black, font=\small] at (mk) {\(m_k\)};
    \draw[dashed, blue!40, line width=0.6pt] (1.5,0) -- (1.5,0.7);
    \node[above, align=center, blue!80!black, font=\footnotesize] at (1.5,0.7)
        {Max fixed point \(< y\)};
    
    \draw[cyan!70!black] plot[mark=x, mark size=4] coordinates {(4.5,0)};
    \node[] (y) at (4.5,0) {};
    \node[below=3pt, cyan!70!black, font=\small] at (y) {\(y\)};
    \draw[dashed, cyan!70!black, line width=0.6pt] (4.5,0) -- (4.5,0.7);
    \node[above, align=center, cyan!70!black, font=\footnotesize] at (4.5,0.7)
        {Hyp. undercut midpoint\\ \scriptsize{(\(H(y) < y\))}};
    
    \node[fixed_point, fill=green!60!black] (mstar) at (7.5,0) {};
    \node[below=3pt, green!60!black, font=\small, xshift=-4pt] at (mstar) {\(m^-\)};
    \draw[dashed, green!40!black!40, line width=0.6pt] (7.5,0) -- (7.5,0.7);
    \node[above, align=center, green!60!black, font=\footnotesize] at (7.5,0.7)
        {PSNE-candidate fixed point\\ \scriptsize{(Bid-raise-immune)}};
    
    \draw[orange!90!black] plot[mark=x, mark size=3pt] coordinates {(2.5,0)}
        node[below=3pt, orange!90!black, font=\small] {\(x\)};
    
    \draw[->, orange!90!black, line width=0.7pt, bend left=60]
        (2.6, 0.1) to
        node[below=-1pt, font=\scriptsize] {\(H(x)\)}
        (3.5, 0.1);
    
    \draw[->, orange!90!black, line width=0.7pt, bend left=60]
        (4.5, -0.1) to
        node[above=-1pt, font=\scriptsize] {\(H(y)\)}
        (3.5, -0.1);
    
    \fill[orange!90!black] (3.5, 0) circle (1.5pt)
        node[above=0.5pt, orange!90!black, font=\small, xshift=7pt] {\(m'\)};
    
    \draw[purple!90!black] plot[mark=x, mark size=3pt] coordinates {(6.0,0)}
        node[below=3pt, purple!90!black, font=\small] {\(x\)};
    
    \draw[teal!80!black] plot[mark=x, mark size=3pt] coordinates {(8.6,0)}
        node[below=3pt, teal!90!black, font=\small] {\(x\)};
    
    \draw[->, teal!80!black, dashed, line width=0.7pt, bend right=30]
        (8.5, 0.15) to
        node[above=-1pt, font=\scriptsize, align=center, xshift=18pt]
        {Bid-raise contagion}
        (7.6, 0.15);
    
    \draw[brace_style, orange!90!black] (1.6,-0.6) -- (4.4,-0.6) 
        node[midway, below=4pt, case_text, text width=2.4cm] {
            \textbf{Case 1}: \(x < y\)\\[0.5ex]
            \(\Downarrow\)\\
            \color{red}{\textbf{Contradicts}}\\
            \color{red}{\textbf{maximality of \(m_k\)}}
        };
        
    \draw[brace_style, purple!90!black] (4.6,-0.6) -- (7.3,-0.6) 
        node[midway, below=4pt, case_text, text width=2.4cm] {
            \textbf{Case 2}: \(y < x < m^-\)\\[0.5ex]
            \(\Downarrow\)\\
            \color{red}{\textbf{Contradicts}}\\
            \color{red}{\textbf{\(x > y\)}}
        };
    
    \draw[brace_style, teal!90!black] (7.6,-0.6) -- (9.4,-0.6) 
        node[midway, below=4pt, case_text, text width=2.1cm] {
            \textbf{Case 3}: \(x > m^-\)\\[0.5ex]
            \(\Downarrow\)\\
            \color{red}{\textbf{Contradicts}}\\
            \color{red}{\textbf{\(m^- \in \mathcal{F}_B\)}}
        };
\end{tikzpicture}
    \caption{
        Visual breakdown of the three case contradictions in the proof of \Cref{thm:nash_existence}.
    }
    \label{fig:nash_existence}
\end{figure}

\begin{restatable}[Existence of a Pure-Strategy Nash Equilibrium]{theorem}{nashExistence} \label{thm:nash_existence}
    Define
    \[
        \mathcal{F}_B \coloneq \{ m\in\mathcal{F} \mid \text{no profitable bid raise exists at } m \}
    \]
    and \(m^- \coloneq \inf \mathcal{F}_B\). Then \(\mathcal{F}_B\neq\varnothing\), and \(\mathbf{s}(m^-)\) is a pure-strategy Nash equilibrium.
\end{restatable}
\begin{proofsketch}
    Since \(H\) is continuous, \(\mathcal{F}\) is compact: it has a minimum and maximum fixed point, \(m_{\min}\) and \(m_{\max}\), and contains the infimum and supremum of any of its nonempty subsets. \Cref{lem:extremal_immunity} makes \(m_{\max}\) immune to bid raises, so \(\mathcal{F}_B \neq \varnothing\) and \(m^- = \inf \mathcal{F}_B \in \mathcal{F}\). A short contagion argument (\Cref{lem:deviations_contagion}) shows \(m^- \in \mathcal{F}_B\), so \(m^- = \min \mathcal{F}_B\) is itself immune to bid raises.

    It remains to show that \(m^-\) is also immune to undercuts. Suppose, towards a contradiction, that a profitable undercut at \(m^-\) induces \(y<m^-\). \Cref{lem:extremal_immunity,lem:deviations_contagion} place \(y\) strictly between \(m_{\min}\) and \(m^-\), and compactness gives a largest fixed point \(m_k<y\). Since \(m_k < m^- = \min \mathcal{F}_B\), \(m_k\) admits a profitable bid raise to some \(x>m_k\); \Cref{lem:implications_deviations} gives \(H(x) > x\), which rules out \(x = y\) and \(x = m^-\), leaving three cases for \(x\):
    \begin{enumerate}
        \item Case \(x < y\): \(H\) then maps \([x,y]\) into itself, so Knaster--Tarski produces a fixed point strictly between \(x\) and \(y\), contradicting the maximality of \(m_k\).
        
        \item Case \(y < x < m^-\): Inequalities from \Cref{lem:implications_deviations} combine, via monotonicity of \(a(\cdot), b(\cdot)\), into \(x<y\), contradicting \(y<x\).
            
        \item Case \(x > m^-\): \Cref{lem:deviations_contagion} carries the bid raise at \(m_k\) forward to \(m^-\), contradicting \(m^- \in \mathcal{F}_B\).
    \end{enumerate}
    No case survives, so \(m^-\) is also undercut-immune, and \Cref{lem:single_sided_deviations} concludes. The full proof is given in \Cref{app:proofs_equilibrium}.
\end{proofsketch}

While the existence result above focuses on the minimal fixed point immune to bid raises, a natural dual argument applies symmetrically to undercuts.

\begin{restatable}[Dual Existence of a Pure-Strategy Nash Equilibrium]{theorem}{nashExistenceDuality} \label{thm:nash_existence_duality}
    Define
    \[
        \mathcal{F}_A\coloneq \{m\in\mathcal{F} \mid \text{no profitable aks undercut exists at } m\}
    \]
    and \(m^+ \coloneq \sup \mathcal{F}_A\). Then \(\mathcal{F}_A\neq\varnothing\), and \(\mathbf{s}(m^+)\) is a pure-strategy Nash equilibrium.
\end{restatable}

\Cref{thm:nash_existence,thm:nash_existence_duality} thus yield two, possibly coinciding, equilibrium mid-prices, \(m^-\) and \(m^+\), obtained from dual constructions. We now show that these are not just any two equilibria: they bracket every other one.

\begin{restatable}[Structure of Equilibria]{theorem}{nashStructure} \label{thm:nash_structure}
    The equilibrium midpoints \(m^-\) and \(m^+\) satisfy \(m^- \le m^+\), and are respectively the smallest and largest equilibrium midpoints. Hence, every equilibrium mid-price belongs to the interval \([m^-,m^+]\).
\end{restatable}
\begin{proof}
    Any PSNE mid-price \(m^*\) admits neither a profitable ask undercut nor a profitable bid raise, so \(m^*\in\mathcal{F}_A\cap\mathcal{F}_B\) and hence \(m^-=\inf\mathcal{F}_B\le m^*\le\sup\mathcal{F}_A=m^+\). Since \(m^-\) and \(m^+\) are themselves PSNE mid-prices (\Cref{thm:nash_existence,thm:nash_existence_duality}), taking \(m^*=m^-\) yields \(m^-\le m^+\), and the bounds show that they are the smallest and largest PSNE mid-prices, respectively.
\end{proof}

Concretely, combining all preceding results, for any equilibrium mid-price \(m^* \in [m^-, m^+]\), let \(a(m^*)\) and \(b(m^*)\) be the corresponding zero-profit quotes. Then, any profile \(\mathbf{s} = (\mathbf{b}, \mathbf{a})\) in which
\begin{enumerate}
    \item at least two market makers quote the best ask \(a(m^*)\) and at least two market makers quote the best bid \(b(m^*)\) (ensuring \(n_A(\mathbf{a}) \ge 2\) and \(n_B(\mathbf{b}) \ge 2\)), and
    
    \item all remaining market makers \(k\) submit non-improving quotes, i.e., \(a_k \ge a(m^*)\) and \(b_k \le b(m^*)\),
\end{enumerate}
constitutes a pure-strategy Nash equilibrium. Thus, equilibrium is robust to any passive configuration of non-competitively quoting market makers.

\begin{figure}[t]
    \centering
    \begin{tikzpicture}[
    x=0.75\linewidth, y=1cm,
    font=\scriptsize\sffamily,
    >={Stealth[length=5pt,width=4pt]}
]

    \colorlet{eq}{blue!75!black}
    \colorlet{dev}{red!70!black}

    \def\vmin{0.25}   
    \def\vmax{1.00}   
    \def\mstar{0.12}  
    \def\astar{0.75}  
    \def\bprime{0.25} 
    \def\aprime{0.85} 
    \def\y{0.50}      

    \fill[orange!20] (\vmin, -0.05) rectangle (\y, 0.05);

    \draw[thick, black] (\mstar-0.08, 0) -- (\vmax+0.05, 0) node[right, font=\bfseries\scriptsize] {$V$};

    \draw[thick] (\vmin, -0.1) -- (\vmin, 0.1) node[above] {$\underline{v}$};
    \draw[thick] (\vmax, -0.1) -- (\vmax, 0.1) node[above] {$\overline{v}$};

    \node[eq, anchor=south west, font=\bfseries\footnotesize] at (\mstar-0.08, 2.3) {Out-of-Support Zero-Profit Profile ($\widehat{m} \le \underline{v}$)};

    \draw[eq, dashed] (\mstar, 0) -- (\mstar, 1.8) node[above] {$\widehat{m}$};
    \draw[eq, dashed] (\astar, 0) -- (\astar, 0) node[above] {$a(\widehat{m})$};

    \fill[eq] (\mstar, 0) circle (1.5pt);
    \fill[eq] (\astar, 0) circle (1.5pt);

    \node[black!60, font=\tiny, align=center, rotate=90] at (\mstar-0.04, 1) {Inactive Bid Side};

    \draw[->, eq, thick] (\vmin+0.02, 0.8) -- (\vmax-0.02, 0.8) node[midway, fill=white, inner sep=2pt, font=\scriptsize] {100\% Buy Orders};

    \draw[decorate, decoration={brace, amplitude=5pt, mirror}, eq, thick] 
        (\vmax-0.01, 1.4) -- (\vmin+0.01, 1.4) 
        node[midway, above=2pt] {$u^\mathrm{ask}_i(\mathbf{s}(\widehat{m})) = 0$};

    \node[dev, anchor=north west, font=\bfseries\footnotesize] at (\mstar-0.08, -2.3) {Unilateral Deviation ($\widetilde{s}_i = (\widetilde{b}_i, \widetilde{a}_i)$)};

    \draw[dev, dashed] (\bprime, 0) -- (\bprime, 0) node[below] {$\widetilde{b}_i$};
    \draw[black!60, dashed] (\aprime, 0) -- (\aprime, 0) node[below] {$\widetilde{a}_i$};
    \draw[dev, dashed] (\y, 0) -- (\y, -1.8) node[below] {$y$};

    \fill[dev] (\bprime, 0) circle (1.5pt);
    \fill[black!60] (\aprime, 0) circle (1.5pt);
    \fill[dev] (\y, 0) circle (1.5pt);

    \draw[->, thick, orange!90!black, bend right=20] (\mstar, -0.15) to node[below, font=\bfseries\tiny, inner sep=1.5pt] {Midpoint Shift} (\y, -0.15);

    \draw[<-, dev, thick] (\vmin+0.02, -1) -- (\y-0.02, -1) node[midway, fill=white, inner sep=1.5pt, text=dev, font=\tiny] {Restored Sell Flow};
    \draw[->, black!60, thick] (\y+0.02, -1) -- (\vmax-0.02, -1) node[midway, fill=white, inner sep=1.5pt, text=black!60, font=\tiny] {Remaining Buy Flow};

    \draw[decorate, decoration={brace, amplitude=5pt}, dev, thick] 
        (\y-0.01, -1.4) -- (\vmin+0.01, -1.4) 
        node[midway, below=2pt, align=center] {$u^\mathrm{bid}_i(\widetilde{s}_i, \mathbf{s}_{-i}(\widehat{m})) > 0$};

    \draw[decorate, decoration={brace, amplitude=5pt}, black!60, thick] 
        (\vmax-0.01, -1.4) -- (\y+0.01, -1.4) 
        node[midway, below=2pt, align=center, text=black!60] {$u^\mathrm{ask}_i(\widetilde{s}_i, \mathbf{s}_{-i}(\widehat{m})) = 0$};
\end{tikzpicture}
    \caption{
        Profitable boundary deviation at an out-of-support candidate midpoint \(\widehat{m} \le \underline{v}\). Posting \(\widetilde{b}_i = \underline{v}\) while withdrawing the ask shifts the induced midpoint to \(y > \underline{v}\), reactivating sell order flow on \([\underline{v}, y)\) with a strictly positive expected margin.
    }
    \label{fig:out_of_support_deviation}
\end{figure}
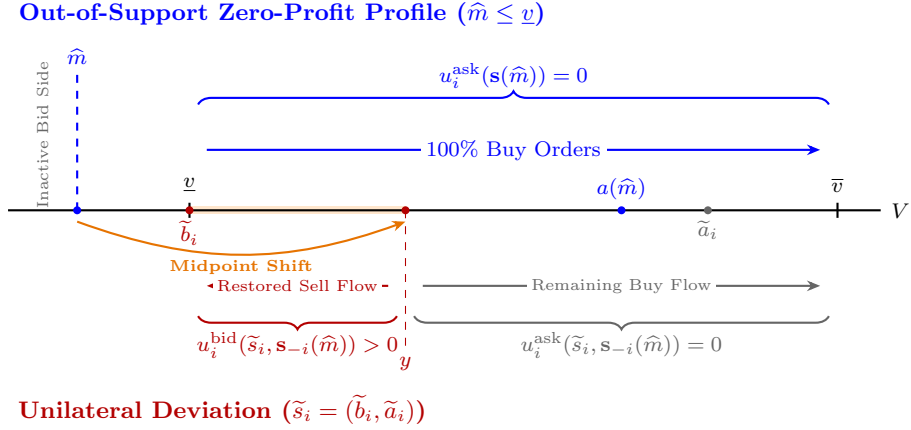

\paragraph{Out-of-Support Midpoints.}
Throughout our analysis, we restricted attention to interior midpoints \(m \in (\underline{v}, \overline{v})\), where both sides of the market are active with positive probability. Although candidate midpoints outside the support (\(\widehat{m} \le \underline{v}\) or \(\widehat{m} \ge \overline{v}\)) can exist, they are economically less interesting and can never sustain a PSNE. Indeed, as illustrated in \Cref{fig:out_of_support_deviation}, if \(\widehat{m} \le \underline{v}\), the bid side is completely inactive; a market maker can unilaterally quote \(\widetilde{b}_i = \underline{v}\) while withdrawing its ask, shifting the induced midpoint into the support and capturing sell orders with strictly positive expected profit, since \(V > \underline{v}\) almost surely. The case \(\widehat{m} \ge \overline{v}\) is symmetric, with \(\widetilde{a}_i = \overline{v}\). Therefore, we can restrict the search for PSNE mid-prices to \((\underline{v}, \overline{v})\).

\paragraph{Beyond Pure-Strategy Nash Equilibria.}
Extending our analysis to richer solution concepts, such as mixed-strategy Nash equilibria (MSNEs) and correlated equilibria (CEs), is a natural and interesting next step. Randomization and correlation can readily be layered onto our results without generating new economic outcomes:
\begin{itemize}
    \item \textbf{Mixed Strategies:} Passive makers, those quoting strictly non-competitively (\(a_i > a(m^*)\) and \(b_i < b(m^*)\)), can randomize freely, because they remain non-executing, so the prevailing best quotes remain bound to the zero-profit competitive threshold.
    \item \textbf{Correlated Strategies:} A public correlation device can synchronize agents on one of the (possibly multiple) PSNEs of \Cref{thm:nash_structure}, again leaving the zero-profit structure intact.
\end{itemize}
What we cannot show in general is the converse: that these are the \emph{only} routes to an MSNE or CE. The obstacle lies specifically in randomization: against a genuinely mixing opponent, a deviation's profitability becomes an expectation over its randomized quotes, whose sign is hard to control without further structure on \(F\). Indeed, under the regularity of \Cref{sec:learning} and uniqueness of the PSNE mid-price, we establish an asymptotic version of this characterization, but achieving it without such assumptions remains a nontrivial open problem.


\subsection{Impact of Asset Variance on Equilibrium Spreads} \label{subsec:asset_variance}

Having established the existence and structure of pure-strategy Nash equilibria, we now examine how the equilibrium bid--ask spread responds to changes in the uncertainty surrounding the underlying asset value \(V\). Intuitively, higher volatility expands the conditional expectations of extreme value realizations, increasing adverse selection risk and forcing makers to widen their quotes.

To formalize this intuition, consider a location-scale family of asset distributions generated by a standardized continuous random variable \(X\) supported on a bounded interval \([\underline{x}, \overline{x}]\), with zero mean \(\mathbb{E}[X]=0\) and unit variance \(\mathrm{Var}(X)=1\). For any location parameter \(\mu \in \mathbb{R}\) and scale parameter \(\sigma > 0\), define the fundamental value as
\[
    V_\sigma \coloneq \mu + \sigma X.
\]
Thus, \(V_\sigma\) has support \([\underline{v}_\sigma, \overline{v}_\sigma] = [\mu + \sigma \underline{x}, \mu + \sigma \overline{x}]\), which we require to lie in \([0, 1]\), with \(\mathbb{E}[V_\sigma] = \mu\) and \(\mathrm{Var}(V_\sigma) = \sigma^2\).
For any mid-price \(m \in (\underline{v}_\sigma, \overline{v}_\sigma)\), an affine change of variable gives the zero-profit ask and bid functions for \(V_\sigma\) as
\begin{align*}
    a_\sigma(m) 
    &\coloneq \mathbb{E}\left[V_\sigma \;\middle|\; V_\sigma > m\right] 
    = \mu + \sigma \,\mathbb{E}\left[X \;\middle|\; X > \tfrac{m-\mu}{\sigma}\right], \\
    b_\sigma(m) 
    &\coloneq \mathbb{E}\left[V_\sigma \;\middle|\; V_\sigma < m\right] 
    = \mu + \sigma \,\mathbb{E}\left[X \;\middle|\; X < \tfrac{m-\mu}{\sigma}\right].
\end{align*}
Hence the midpoint operator \(H_\sigma(m) \coloneq (a_\sigma(m)+b_\sigma(m))/2\) satisfies
\[
    H_\sigma(m) = \mu + \sigma H_1\left( \frac{m-\mu}{\sigma} \right),
\]
where \(H_1(z) \coloneq (\mathbb{E}[X \mid X>z]+\mathbb{E}[X\mid X<z] )/2\) is the normalized midpoint operator for \(X\), extended to \([\underline{x},\overline{x}]\) as in \Cref{def:midpoint_operator} so that \(H_1\) (hence \(H_\sigma\)) is a well-defined self-map of the closed interval.

\begin{proposition}[Scaling of the Zero-Profit Spread] \label{prop:spread_volatility}
    Let \(\mathcal{F}_1 \subset (\underline{x}, \overline{x})\) denote the set of fixed points of \(H_1\), and define the spread functions
    \begin{align*}
        w_1(z)
        &\coloneq \mathbb{E}\left[ X \;\middle|\; X > z \right] - \mathbb{E}\left[ X \;\middle|\; X < z \right],\\
        w_\sigma(m)
        &\coloneq a_\sigma(m) - b_\sigma(m).
    \end{align*}
    For any \(\sigma > 0\):
    \begin{enumerate}[label=(\roman*)]
        \item The set of fixed points of \(H_\sigma\) satisfies \(\mathcal{F}_\sigma = \mu + \sigma \mathcal{F}_1\).
        \item For any \(\widehat{m}_\sigma \in \mathcal{F}_\sigma\), writing \(\widehat{m}_\sigma = \mu + \sigma z\) with \(z \in \mathcal{F}_1\), the spread satisfies
        \[
            w_\sigma(\widehat{m}_\sigma) = \sigma \cdot w_1(z).
        \]
    \end{enumerate}
\end{proposition}
\begin{proof}
    A point \(\widehat{m}_\sigma \in [\underline{v}_\sigma, \overline{v}_\sigma]\) is a fixed point of \(H_\sigma\) if and only if \(\widehat{m}_\sigma = H_\sigma(\widehat{m}_\sigma)\). Substituting the scaled expression for \(H_\sigma\) and setting \(z \coloneq (\widehat{m}_\sigma - \mu) / \sigma\), this is equivalent to
    \[
        \mu + \sigma z = \mu + \sigma H_1(z)
        \iff
        z = H_1(z).
    \]
    Thus, \(\widehat{m}_\sigma \in \mathcal{F}_\sigma\) if and only if \(z \in \mathcal{F}_1\), establishing part~(i).

    For part~(ii), evaluate the spread at \(\widehat{m}_\sigma = \mu + \sigma z\):
    \begin{align*}
        w_\sigma(\widehat{m}_\sigma)
        &= a_\sigma(\widehat{m}_\sigma) - b_\sigma(\widehat{m}_\sigma) \\
        &= \left( \mu + \sigma \, \mathbb{E}\left[ X \;\middle|\; X > z \right] \right) - \left( \mu + \sigma \, \mathbb{E}\left[ X \;\middle|\; X < z \right] \right) \\
        &= \sigma \left( \mathbb{E}\left[ X \;\middle|\; X > z \right] - \mathbb{E}\left[ X \;\middle|\; X < z \right] \right)
        = \sigma \cdot w_1(z),
    \end{align*}
    which concludes the proof.
\end{proof}

Since every PSNE mid-price is, in particular, a fixed point of the corresponding midpoint operator (\Cref{sec:equilibrium}), \Cref{prop:spread_volatility} applies to it directly: the equilibrium bid--ask spread at any PSNE \(\mu + \sigma z\) of the \(V_\sigma\)-market equals \(\sigma w_1(z)\), growing linearly in \(\sigma\).

To put this into perspective, while in information-based models, like \citet{glosten1985bid}, the bid--ask spread widens primarily as a defense mechanism against a higher proportion of informed traders, here the spread widens strictly due to the amplification of fundamental risk \(\sigma\), even in a framework where the underlying information structure remains fixed. Since every competing maker responds to the same volatility shock in the same way, this widening is synchronized across the market, producing exactly the kind of correlated price movement that recent tests for algorithmic collusion are designed to catch. \citet{hartline2024regulation} propose flagging collusion directly from the sequence of played prices, while \citet{hartline2025regulation} show that such price-based tests can be confounded whenever non-collusive forces already drive agents toward the same correlated play. Our result supplies exactly one such force: volatility-driven spread widening is a benign, competitive source of precisely the price patterns these tests are built to flag.


\section{Learning} \label{sec:learning}

We ask whether adaptive makers can recover the competitive price structure without knowing their payoff functions. Continuous-action bandit guarantees typically require smooth payoffs~\citep{cesa2023repeated, cesa2025market}, whereas best-quote allocation makes ours discontinuous. We therefore restrict quotes to a price grid with tick size \(\Delta\), giving rise to the \emph{grid game} \(\mathcal{G}_\Delta\), and show that the market prices of its correlated equilibria, and hence the time-averaged market prices of no-internal-regret learners, approach the Nash price envelope of the continuous game as \(\Delta \to 0\).

For simplicity, we focus on two makers and normalize the support of \(V\) to \([0, 1]\), but neither restriction is essential.\footnote{The argument extends from two makers to any fixed number by replacing the opponent's quote with the best rival quote and tracking ties, and from \([0, 1]\) to any bounded support by an affine rescaling.} As anticipated in~\Cref{sec:equilibrium}, this analysis requires additional regularity.

\begin{assumption}[Regularity] \label{ass:learning}
    The random variable \(V\) has a continuous density \(f > 0\) on \((0, 1)\), and the zero-profit quotes of \Cref{def:midpoint_operator} satisfy \(a'(m) + b'(m) \le 2\) for every \(m \in (0, 1)\).
\end{assumption}

This assumption holds throughout the rest of this section. The first condition makes the zero-profit quotes continuously differentiable, with derivatives
\[
    b'(m) = \dfrac{f(m)}{F(m)} (m - b(m)),
    \qquad
    a'(m) = \dfrac{f(m)}{1 - F(m)} (a(m) - m),
\]
and it also ensures \(b'(m), a'(m) > 0\), making the zero-profit quotes strictly increasing. The second condition bounds how fast \(H\) can move: \(H'(m) = (a'(m) + b'(m))/2 \le 1\), so \(G(m) \coloneq m - H(m)\) has \(G'(m) \ge 0\) and is nondecreasing.

\begin{restatable}[Structure of Equilibria under Regularity]{proposition}{nashPriceInterval} \label{prop:nash_structure_regularity}
    (i) The fixed-point set \(\mathcal{F}\) is a nonempty closed interval \([m^-, m^+]\); and (ii) every \(m^* \in \mathcal{F}\) is a PSNE mid-price, with market prices \((b(m^*), a(m^*))\).
\end{restatable}
\begin{proofsketch}
    For (i), \(\mathcal{F}\) is the zero set of \(G\), which is continuous and nondecreasing; \(\mathcal{F}\) is nonempty by \Cref{thm:existence_zero_profit_midpoint}, hence it is a closed interval. For (ii), at any \(m^* \in \mathcal{F}\), the same monotonicity, driven by \(a'(m) + b'(m) \le 2\), rules out both one-sided and two-sided deviations from \((b(m^*), a(m^*))\), so \(m^*\) is a PSNE mid-price. The full proof is given in \Cref{app:proofs_learning}.
\end{proofsketch}

Define \(b^\pm \coloneq b(m^\pm)\) and \(a^\pm \coloneq a(m^\pm)\). As \(m^*\) ranges over \(\mathcal{F} = [m^-, m^+]\), the market prices \((b(m^*), a(m^*))\) trace out the Nash price pairs, a monotone curve from \((b^-, a^-)\) to \((b^+, a^+)\), whose axis-aligned envelope we call the \emph{Nash price envelope}:
\[
    \mathcal{R}^{\mathrm{NE}} \coloneq [b^-, b^+] \times [a^-, a^+].
\]


\subsection{Zero-Profit Responses}

While \(a(\cdot)\) and \(b(\cdot)\) tie both zero-profit quotes to a common mid-price, we now fix the market quote on one side and ask which quote on the other side breaks even.

\begin{definition}[Zero-Profit Responses] \label{def:zero_profit_responses}
    For any market bid \(B \in [0, 1)\), \(\alpha(B)\) is the unique solution \(x \in (B, 1)\) of \(x = a((x + B)/2)\). Symmetrically, for any market ask \(A \in (0, 1]\), \(\beta(A)\) is the unique solution \(x \in (0, A)\) of \(x = b((A + x)/2)\). We extend them by \(\alpha(1) \coloneq 1\) and \(\beta(0) \coloneq 0\), consistently with their limits.
\end{definition}

In words, \(\alpha(B)\) is the only market ask that earns zero expected profit against the market bid \(B\), given the mid-price it induces; \(\beta(A)\) is its bid-side counterpart. Both are well defined, since \(x \mapsto x - a((x + B)/2)\) is strictly increasing, negative at \(x = B\) and positive at \(x = 1\), and symmetrically for \(\beta(A)\).\footnote{Its derivative \(1 - a'((x + B)/2)/2\) is positive, as \(a'(m) \le 2 - b'(m) < 2\); the signs follow from \(m < a(m) < 1\) for \(m \in [0, 1)\).} By the implicit function theorem, they are \(C^1\) on \([0, 1)\) and \((0, 1]\), respectively, with
\[
    \alpha'(B) = \dfrac{a'(m)}{2 - a'(m)}
    \qquad \text{and} \qquad
    \beta'(A) = \dfrac{b'(m)}{2 - b'(m)}
\]
at the induced mid-prices \(m = (B + \alpha(B))/2\) and \(m = (A + \beta(A))/2\). Hence, both are strictly increasing and Lipschitz on every compact subinterval of these intervals.

At any \(m^* \in \mathcal{F}\), since \(m^* = H(m^*)\), the responses recover the equilibrium quotes: \(\alpha(b(m^*)) = a(m^*)\) and \(\beta(a(m^*)) = b(m^*)\), and in particular \(\alpha(b^\pm) = a^\pm\) and \(\beta(a^\pm) = b^\pm\). Composing them yields the following property, proved in \Cref{app:proofs_learning}.

\begin{restatable}{lemma}{noCycle} \label{lem:no_cycle}
    For every market ask \(A \in (0, 1]\),
    \[
        \alpha(\beta(A))
        \begin{cases}
            > A & \text{if } A < a^-, \\
            = A & \text{if } A \in [a^-, a^+], \\
            < A & \text{if } A > a^+.
        \end{cases}
    \]
\end{restatable}

That is, if the bid side responds to an ask \(A\) with \(\beta(A)\), and the ask side responds in turn with \(\alpha(\beta(A))\), the ask stays at \(A\) when \(A \in [a^-, a^+]\) and otherwise moves strictly toward this interval; the same holds for bids, with \(\beta \circ \alpha\) and \([b^-, b^+]\).


\subsection{The Grid Game}

Fix a tick size \(\Delta > 0\) with \(1/\Delta \in \mathbb{N}\), and let \(\Gamma_\Delta \coloneq \{0, \Delta, 2\Delta, \dots, 1\}\). The grid game \(\mathcal{G}_\Delta\) is the market-making game of \Cref{sec:model} with each maker's strategy space restricted to
\[
    \mathcal{S}_i^\Delta
    \coloneq \left\{ (b, a) \in \Gamma_\Delta^2 \;\middle|\; b \le a \right\} \subset \mathcal{S}_i
    \qquad \text{and} \qquad
    \mathcal{S}^\Delta \coloneq \prod\nolimits_{i \in \mathcal{N}} \mathcal{S}_i^\Delta,
\]
so that best quotes, mid-prices and payoffs are exactly as defined there. The next two lemmas identify profitable deviations on the grid: a best ask more than one tick below \(\alpha(B)\) is profitably raised by one tick, while a rival ask sufficiently many ticks above \(\alpha(B)\) is profitably undercut by one tick; symmetrically for bids.

\begin{restatable}{lemma}{aggressiveTick} \label{lem:aggressive_tick}
    Let \(\mathbf{s}=(\mathbf{b}, \mathbf{a})  \in \mathcal{S}^\Delta\), \(A \coloneq a^*(\mathbf{a})\), and \(B \coloneq b^*(\mathbf{b})\). If \(a_i = A\) and \(a_i + \Delta < \alpha(B)\), then raising \(a_i\) to \(a_i + \Delta\) strictly increases \(u_i\). Symmetrically, if \(b_i = B\) and \(b_i - \Delta > \beta(A)\), then lowering \(b_i\) to \(b_i - \Delta\) strictly increases \(u_i\).
\end{restatable}

\begin{restatable}{lemma}{passiveTick} \label{lem:passive_tick}
    Fix \(\overline{B} < 1\) and \(\underline{A} > 0\). There exists \(C_0 > 1\), depending only on \(F\), \(\overline{B}\) and \(\underline{A}\), such that the following holds. Let \(\mathbf{s}=(\mathbf{b}, \mathbf{a}) \in \mathcal{S}^\Delta\), \(A \coloneq a^*(\mathbf{a})\), \(B \coloneq b^*(\mathbf{b})\), and let \(i\) be a maker. If \(B \le \overline{B}\) and \(\widetilde{a} \in \Gamma_\Delta\) satisfies \(a_{-i} \le \widetilde{a} \le a_i\) and \(\widetilde{a} \ge \alpha(B) + C_0 \Delta\), then lowering \(a_i\) to \(\widetilde{a} - \Delta\) weakly increases \(u_i\), strictly if \(a_{-i} \ge \widetilde{a} - \Delta\). Symmetrically, if \(A \ge \underline{A}\) and \(\widetilde{b} \in \Gamma_\Delta\) satisfies \(b_i \le \widetilde{b} \le b_{-i}\) and \(\widetilde{b} \le \beta(A) - C_0 \Delta\), then raising \(b_i\) to \(\widetilde{b} + \Delta\) weakly increases \(u_i\), strictly if \(b_{-i} \le \widetilde{b} + \Delta\).
\end{restatable}


\subsection{Correlated Equilibria}

Since the time-averaged play of no-internal-regret learners approaches the set of correlated equilibria~\citep{foster1997calibrated}, we study the latter.

\begin{definition}[Correlated Equilibrium] \label{def:ce}
    For \(\delta \ge 0\), a distribution \(\pi\) over \(\mathcal{S}^\Delta\) is a \emph{\(\delta\)-correlated equilibrium} (\(\delta\)-CE) of \(\mathcal{G}_\Delta\) if, for every maker \(i \in \mathcal{N}\) and every pair \(s_i, \widetilde{s}_i \in \mathcal{S}_i^\Delta\),
    \[
        \sum\nolimits_{\mathbf{s}_{-i}} \pi(s_i, \mathbf{s}_{-i}) \left( u_i(\widetilde{s}_i, \mathbf{s}_{-i}) - u_i(s_i, \mathbf{s}_{-i}) \right) \le \delta.
    \]
    A \(0\)-CE is simply called a \emph{correlated equilibrium} (CE).
\end{definition}

For a CE \(\pi\), let
\[
    a^{\min}(\pi) \coloneq \min_{\mathclap{\mathbf{s} \in \supp(\pi)}} \, a^*(\mathbf{a})
    \qquad \text{and} \qquad
    a^{\max}(\pi) \coloneq \max_{\mathclap{\mathbf{s} \in \supp(\pi)}} \, a^*(\mathbf{a})
\]
denote the extremal market asks on its support, and define \(b^{\min}(\pi)\) and \(b^{\max}(\pi)\) similarly.

\begin{restatable}{lemma}{cePriceBounds} \label{lem:ce_price_bounds}
    There exist \(C_0 > 1\) and \(\Delta_0 > 0\), depending only on \(F\), such that, for every \(\Delta \le \Delta_0\), every CE \(\pi\) of \(\mathcal{G}_\Delta\) satisfies
    \begin{align*}
        a^{\min}(\pi) &\ge \alpha(b^{\min}(\pi)) - \Delta, &
        b^{\max}(\pi) &\le \beta(a^{\max}(\pi)) + \Delta, \\
        a^{\max}(\pi) &\le \alpha(b^{\max}(\pi)) + C_0 \Delta, &
        b^{\min}(\pi) &\ge \beta(a^{\min}(\pi)) - C_0 \Delta.
    \end{align*}
\end{restatable}

In words, up to one tick below and \(C_0\) ticks above, the market asks of any CE are bracketed by the zero-profit responses to its extremal bids; symmetrically for bids.


\subsection{Grid Correlated Equilibria Approach Nash}

For a profile \(\mathbf{s}\), let \(D_{\mathrm{env}}(\mathbf{s}) \coloneq \min_{(x, y) \in \mathcal{R}^{\mathrm{NE}}} (|b^*(\mathbf{b}) - x| + |a^*(\mathbf{a}) - y|)\) be the \(\ell_1\)-distance between its market prices and the Nash price envelope, and define the grid approximation error
\[
    \varepsilon_\Delta \coloneq \sup_{\pi} \max_{\mathbf{s} \in \supp(\pi)} D_{\mathrm{env}}(\mathbf{s}),
\]
where \(\pi\) ranges over all CEs of \(\mathcal{G}_\Delta\).

\begin{restatable}[Convergence of Grid Correlated Equilibria]{theorem}{gridCE} \label{thm:grid_ce}
    \(\varepsilon_\Delta \to 0\) as \(\Delta \to 0\).
\end{restatable}
\begin{proofsketch}
    Combining the bounds of \Cref{lem:ce_price_bounds}, and using that \(\alpha\) is Lipschitz, the extremal asks of any CE satisfy the two-cycle condition \(\alpha(\beta(A)) = A\) up to an error of order \(\Delta\). By \Cref{lem:no_cycle}, this forces their limit points into \([a^-, a^+]\) as \(\Delta \to 0\), and the continuity of \(\beta\) carries the conclusion over to the extremal bids and \([b^-, b^+]\). Since every supported market price lies between the extremal ones, the convergence is uniform over CEs. The full proof is given in \Cref{app:proofs_learning}.
\end{proofsketch}

In particular, if \(\mathcal{F} = \{m^*\}\), then \(\mathcal{R}^{\mathrm{NE}}\) collapses to the PSNE prices \((b(m^*), a(m^*))\), and the market prices of all grid CEs converge uniformly to them.


\subsection{No-Internal-Regret Learning}

\Cref{thm:grid_ce} extends to approximate CEs, and hence to learning dynamics.

\begin{restatable}[Approximate Correlated Equilibria]{lemma}{approxCE} \label{lem:approx_ce}
    Fix \(\Delta > 0\) with \(1/\Delta \in \mathbb{N}\). If \(\pi_n\) is a \(\delta_n\)-CE of \(\mathcal{G}_\Delta\) for every \(n\), with \(\delta_n \to 0\), then
    \[
        \limsup_{n} \mathbb{E}_{\mathbf{s} \sim \pi_n}[D_{\mathrm{env}}(\mathbf{s})] \le \varepsilon_\Delta~.
    \]
\end{restatable}

In the repeated grid game, at each round \(t\), the makers' learning algorithms select a profile \(\mathbf{S}_t = (S_{t,1}, S_{t,2}) \in \mathcal{S}^\Delta\), the value \(V_t \sim F\) is drawn independently of \(\mathbf{S}_t\) and of the past, and each maker observes only its realized payoff \(r_i(\mathbf{S}_t; V_t, d_t)\), where \(d_t\) is the trader's order. We use uppercase letters because both the profile and the value are random: the former through the learners' randomization and its dependence on past feedback. Maker \(i\) has \emph{vanishing internal regret} if, for every pair \(s_i, \widetilde{s}_i \in \mathcal{S}_i^\Delta\),
\[
    R_{T,i}(s_i, \widetilde{s}_i) \coloneq \frac{1}{T} \sum_{t=1}^{T} \mathbb{E}\left[ \left( u_i(\widetilde{s}_i, \mathbf{S}_{t,-i}) - u_i(\mathbf{S}_t) \right) \mathbb{I}_{\{S_{t,i} = s_i\}} \right]
\]
satisfies \(\limsup_{T} R_{T,i}(s_i, \widetilde{s}_i) \le 0\), as guaranteed, e.g., by regret-matching procedures~\citep{hart2000simple} and swap-regret minimizers~\citep{blum2007external}.

\begin{corollary}[No-Internal-Regret Learning] \label{cor:no_internal_regret_learning}
    Fix \(\Delta > 0\) with \(1/\Delta \in \mathbb{N}\). If both makers have vanishing internal regret, then
    \[
        \limsup_{T \to \infty} \frac{1}{T} \sum_{t=1}^{T} \mathbb{E}\left[ D_{\mathrm{env}}(\mathbf{S}_t) \right] \le \varepsilon_\Delta,
    \]
    and hence, by \Cref{thm:grid_ce}, the left-hand side vanishes as \(\Delta \to 0\).
\end{corollary}
\begin{proof}
    Since \(V_t\) is independent of \(\mathbf{S}_t\), \(R_{T,i}\) coincides with the internal regret on realized payoffs. Let \(\bar{\pi}_T\) be the average of the laws of \(\mathbf{S}_1, \dots, \mathbf{S}_T\). Each sum in \Cref{def:ce} for \(\bar{\pi}_T\) equals the corresponding \(R_{T,i}(s_i, \widetilde{s}_i)\), so \(\bar{\pi}_T\) is a \(\delta_T\)-CE, with \(\delta_T \to 0\) the largest positive part of these regrets. By linearity, \(\mathbb{E}_{\mathbf{s} \sim \bar{\pi}_T}[D_{\mathrm{env}}(\mathbf{s})]\) is the time average in the statement, and \Cref{lem:approx_ce} concludes.
\end{proof}

To summarize, on fine grids every CE, and hence the time-averaged play of no-internal-regret learners, keeps market prices close to the envelope \(\mathcal{R}^{\mathrm{NE}}\). The envelope is only an outer bound: our results do not exclude that learners occasionally play price pairs away from the \emph{Nash curve}, such as \((b^-, a^+)\), but they cannot settle on one. Indeed, if their time-averaged play converges to a single profile, the point mass on it is a CE of \(\mathcal{G}_\Delta\), whose extremal prices coincide, so \Cref{lem:ce_price_bounds} gives \(a^*(\mathbf{a}) \approx \alpha(b^*(\mathbf{b}))\) and \(b^*(\mathbf{b}) \approx \beta(a^*(\mathbf{a}))\), and \Cref{lem:no_cycle} forces its market prices close to the Nash curve.


\section{Simulations} \label{sec:simulations}

\begin{figure}[t]
    \centering
    \begin{subfigure}{0.4\columnwidth}
        \includegraphics[width=\linewidth]{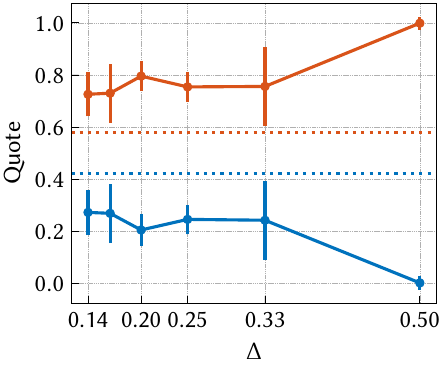}
    \end{subfigure}
    \hspace{.5cm}
    \begin{subfigure}{0.4\columnwidth}
        \includegraphics[width=\linewidth]{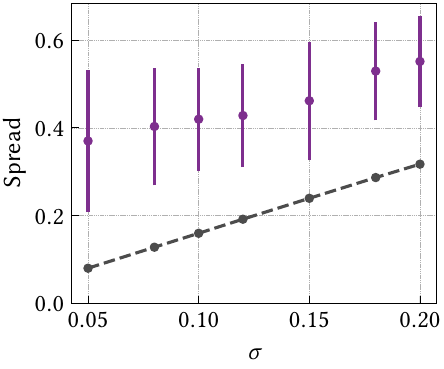}
    \end{subfigure}
    \Description{
        Two panels showing simulations of two learning market makers. In the left panel, as the tick size shrinks, the learned market ask and bid move toward the continuous Nash price envelope shown as dotted reference lines, without fully reaching it at the finest grid tested. In the right panel, at a fixed tick size, the learned spread increases monotonically with the standard deviation of the asset value, tracking the shape of the theoretical spread but sitting consistently above it.
    }
    \caption{
        Simulations of two learning market makers running swap-regret minimizers.
        \emph{Left:} As \(\Delta\) shrinks, the learned market ask (orange) and bid (blue) move toward the continuous Nash envelope (dotted lines), see \Cref{cor:no_internal_regret_learning}.
        \emph{Right:} At a fixed tick size, the learned spread (purple) grows monotonically with the standard deviation \(\sigma\) of the asset value, tracking the theoretical spread (black) while sitting consistently above it, see \Cref{prop:spread_volatility}.
        Error bars are the standard deviation of the learned quotes across independent replications.
    }
    \label{fig:simulation}
\end{figure}

We illustrate the theoretical results with two market makers running Exp3 inside the Blum--Mansour reduction to internal regret~\citep{cesa2006prediction, blum2007external}, against a Gaussian asset value with mean \(0.5\) clamped\footnote{The clamped mass is negligible for the values of \(\sigma\) we test.} to \([0,1]\). Every expert uses a learning rate \(\eta = 0.05\).

\Cref{fig:simulation} (left) fixes \(\sigma = 0.1\) and sweeps the tick size \(\Delta\), learning each grid for \(T = 100{,}000\) rounds and averaging over \(100\) runs. As \(\Delta\) shrinks, the learned market bid and ask move toward the Nash price envelope; the non-monotonicity is not surprising, since even the exact equilibria of the grid game need not vary smoothly in \(\Delta\) at this resolution.
\Cref{fig:simulation} (right) fixes \(\Delta = 0.1\) and sweeps \(\sigma\), learning each value for \(T = 1{,}500{,}000\) rounds and averaging over \(15\) runs. The learned spread grows monotonically and close to linearly with \(\sigma\), consistent with \Cref{prop:spread_volatility}.


\section{Conclusion and Future Work} \label{sec:conclusion}

We modeled market making as a game in which makers post bid--ask quotes and an inelastic trader, who knows the asset value, always trades at the best quotes, buying whenever the value exceeds the market mid-price and selling otherwise. In this game, pure-strategy Nash equilibria always exist, and all equilibrium mid-prices lie between two extremal ones. These equilibria also pin down how competitive spreads depend on risk: within location-scale families, they grow linearly with the volatility of the asset. Finally, we proved that, under an additional regularity condition, no-internal-regret learners recover this competitive benchmark, with time-averaged market prices that lie close to the Nash price envelope on a fine price grid.

This benchmark clarifies what competitive pricing looks like in adaptive markets. In particular, when the asset becomes more volatile, the spreads of all makers widen together, with no coordination among them. Since algorithmic-collusion tests are designed to flag exactly this kind of simultaneous widening across makers, they should control for volatility, which is nontrivial in markets with incomplete information, where volatility itself must be estimated. 

Future work includes a finer analysis of learning: we only showed that time-averaged market prices approach the Nash price envelope, leaving open whether learners concentrate on the Nash curve and at what rate. Other directions are to weaken the regularity required in \Cref{sec:learning}, to characterize mixed-strategy and correlated equilibria in the general model, and to extend the model to inventory, latency, heterogeneous order flow and time-varying fundamentals.

\clearpage



\begin{acks}
    LF acknowledges support by the MUR PRIN grant 2022EKNE5K (Learning in Markets and Society), by the FAIR (Future Artificial Intelligence Research) project, funded by the NextGenerationEU program within the PNRR-PE-AI scheme (M4C2, investment 1.3, line on Artificial Intelligence), and by the EU Horizon CL4-2022-HUMAN-02 research and innovation action under grant agreement 101120237, project ELIAS (European Lighthouse of AI for Sustainability).
    
    MO gratefully acknowledges Intesa Sanpaolo Innovation Center for partially funding his doctoral fellowship, as well as the financial support from the EU Horizon CL4-2022-HUMAN-02 research and innovation action under grant agreement 101120237, project ELIAS (European Lighthouse of AI for Sustainability).
\end{acks}



\bibliography{bibliography}

\clearpage


\appendix
\section{Appendix: Omitted Proofs} \label{app:proofs}

In this appendix, we provide the complete proofs of all technical results presented in the main text.


\subsection{Proofs of Section~\ref{sec:equilibrium}} \label{app:proofs_equilibrium}

\singleSidedDeviations*
\begin{proof}
    Let \(\widetilde{s}_i \coloneq (\widetilde{b}_i, \widetilde{a}_i)\) be a profitable deviation against \(\mathbf{s}_{-i}(\widehat{m})\), inducing mid-price \(z\). Since total profit is the sum of ask-side and bid-side profits, at least one side must yield strictly positive expected profit. Suppose first that the ask side is strictly profitable. Then \(\widetilde{a}_i < a(\widehat{m})\): if \(\widetilde{a}_i > a(\widehat{m})\), maker \(i\) captures no buy orders, while if \(\widetilde{a}_i = a(\widehat{m})\), the market ask is unchanged and the market bid can only rise, so \(z \ge \widehat{m}\) and its ask-side payoff \(\frac{1}{N} \mathbb{P}(V > z)(a(\widehat{m}) - a(z))\) is non-positive. Hence, maker \(i\) alone quotes the market ask, the midpoint is
    \[
        z = \dfrac{\widetilde{a}_i + \max(\widetilde{b}_i, b(\widehat{m}))}{2},
    \]
    and
    \[
        u_i^{\mathrm{ask}}(\widetilde{s}_i, \mathbf{s}_{-i}(\widehat{m})) = \mathbb{P}(V > z)\left(\widetilde{a}_i - a(z)\right) > 0,
    \]
    so that, in particular, the margin satisfies \(\widetilde{a}_i - a(z) > 0\).

    Now consider the ask undercut \((0, \widetilde{a}_i)\), which keeps the same ask quote and withdraws the bid. Under this single-sided deviation, the market bid reverts to \(b(\widehat{m})\), shifting the induced midpoint to
    \[
        y \coloneq \dfrac{\widetilde{a}_i + b(\widehat{m})}{2} \le z,
    \]
    since \(b(\widehat{m}) \le \max(\widetilde{b}_i, b(\widehat{m}))\). Given that \(a(\cdot)\) is nondecreasing, \(y \le z\) expands the margin,
    \[
        \widetilde{a}_i - a(y) \ge \widetilde{a}_i - a(z) > 0,
    \]
    as well as the execution probability,
    \[
        \mathbb{P}(V > y) \ge \mathbb{P}(V > z).
    \]
    Since maker \(i\) has withdrawn from the bid side, its bid profit under \((0, \widetilde{a}_i)\) is zero. Total profit under the new deviation therefore satisfies
    \begin{align*}
        u_i((0, \widetilde{a}_i), \mathbf{s}_{-i}(\widehat{m}))
        &= \mathbb{P}(V > y)\left(\widetilde{a}_i - a(y)\right) \\
        &\ge \mathbb{P}(V > z)\left(\widetilde{a}_i - a(z)\right) = u_i^{\mathrm{ask}}(\widetilde{s}_i, \mathbf{s}_{-i}(\widehat{m})) > 0.
    \end{align*}
    Hence, the ask undercut paired with bid withdrawal remains strictly profitable.

    Symmetrically, if the bid side is strictly profitable, then \(\widetilde{b}_i > b(\widehat{m})\), and withdrawing the ask, i.e., playing \((\widetilde{b}_i, 1)\), shifts the midpoint from \(z\) to
    \[
        x \coloneq \dfrac{a(\widehat{m}) + \widetilde{b}_i}{2} \ge z,
    \]
    which again enhances both margin and execution probability.
\end{proof}

\implicationsDeviations*
\begin{proof}
    We first prove (i). Fix \(x > \widehat{m}\). Under a bid raise by maker \(i\) (paired with ask withdrawal), the prevailing market ask remains \(a(\widehat{m})\). The induced midpoint \(x = (a(\widehat{m}) + \widetilde{b}_i)/2\) uniquely yields
    \[
        \widetilde{b}_i = 2x - a(\widehat{m}),
    \]
    which is strictly profitable if and only if
    \[
        b(x) - \widetilde{b}_i > 0 \iff 2x - b(x) < a(\widehat{m}).
    \]
    Substituting \(b(x) = 2H(x) - a(x)\) yields
    \[
        2\left(x - H(x)\right) + a(x) - a(\widehat{m}) < 0.
    \]
    Since \(x > \widehat{m}\) and \(a(\cdot)\) is nondecreasing, we have \(a(x) - a(\widehat{m}) \ge 0\). For the inequality to hold, it must be that
    \[
        2\left(x - H(x)\right) < a(\widehat{m}) - a(x) \le 0,
    \]
    which strictly requires \(H(x) > x\).

    Symmetrically, for (ii), fix \(y < \widehat{m}\). Under an ask undercut by maker \(i\) (paired with bid withdrawal), the prevailing market bid remains \(b(\widehat{m})\). The induced midpoint \(y = (\widetilde{a}_i + b(\widehat{m}))/2\) uniquely determines the required ask quote:
    \[
        \widetilde{a}_i = 2y - b(\widehat{m}).
    \]
    This undercut is strictly profitable if and only if
    \[
        \widetilde{a}_i - a(y) > 0 \iff 2y - a(y) > b(\widehat{m}).
    \]
    Substituting \(a(y) = 2H(y) - b(y)\) rewrites this profitability condition as
    \[
        2\left(y - H(y)\right) + b(y) - b(\widehat{m}) > 0.
    \]
    Since \(y < \widehat{m}\) and \(b(\cdot)\) is nondecreasing, we have \(b(y) - b(\widehat{m}) \le 0\). For the inequality to hold, it must be that
    \[
        2\left(y - H(y)\right) > b(\widehat{m}) - b(y) \ge 0,
    \]
    which strictly requires \(H(y) < y\).
\end{proof}

\extremalImmunity*
\begin{proof}
    By Knaster--Tarski, \(\mathcal{F}\) is itself a complete lattice~\citep{tarski1955lattice}, so \(m_{\min}\) and \(m_{\max}\) are well defined. It remains to rule out the corresponding one-sided deviation at each extremal point, which we do by contradiction.

    We first prove (i). Suppose, by contradiction, that a profitable ask undercut exists at \(m_{\min}\). By \Cref{lem:implications_deviations}, the induced midpoint \(y < m_{\min}\) satisfies \(H(y) < y\). Since \(H\) is nondecreasing and \(H(\underline{v}) > \underline{v}\), for every \(m \in [\underline{v}, y]\) we have
    \[
        \underline{v} < H(\underline{v}) \le H(m) \le H(y) < y.
    \]
    Hence, \(H\) maps the closed interval \([\underline{v}, y]\) into itself. Then, by the Knaster--Tarski fixed-point theorem, \(H\) admits a fixed point \(\check{m} \in [\underline{v}, y]\). Consequently, \(\check{m} \in \mathcal{F}\) satisfies \(\check{m} \le y < m_{\min} = \min \mathcal{F}\), a contradiction.

    Symmetrically, for (ii), suppose a profitable bid raise exists at \(m_{\max}\). By \Cref{lem:implications_deviations}, the induced midpoint \(x > m_{\max}\) satisfies \(H(x) > x\). Since \(H(\overline{v}) < \overline{v}\), for every \(m \in [x, \overline{v}]\) we have
    \[
        x < H(x) \le H(m) \le H(\overline{v}) < \overline{v}.
    \]
    Thus, \(H\) maps the closed interval \([x, \overline{v}]\) into itself, yielding a fixed point \(\check{m} \in [x, \overline{v}]\) by Knaster--Tarski. Consequently, \(\check{m} \in \mathcal{F}\) satisfies \(\check{m} \ge x > m_{\max} = \max \mathcal{F}\), a contradiction.
\end{proof}

\deviationsContagion*
\begin{proof}
    We first prove (i). Consider a profitable bid raise at \(m_1\) inducing a midpoint \(x > m_2\). By \Cref{lem:implications_deviations}, profitability implies
    \[
        2x - b(x) < a(m_1).
    \]
    Since \(m_1 < m_2\) and the function \(a(\cdot)\) is nondecreasing, we have \(a(m_1) \le a(m_2)\), which directly yields
    \[
        2x - b(x) < a(m_2).
    \]
    Since \(x > m_2\), \Cref{lem:implications_deviations} guarantees that the bid raise inducing midpoint \(x\) at \(m_2\) is likewise strictly profitable.

    Symmetrically, for (ii), a profitable ask undercut at \(m_2\) inducing midpoint \(y < m_1\) satisfies
    \[
        2y - a(y) > b(m_2).
    \]
    Since \(b(\cdot)\) is nondecreasing and \(m_1 < m_2\), we have \(b(m_2) \ge b(m_1)\), implying
    \[
        2y - a(y) > b(m_1).
    \]
    Since \(y < m_1\), \Cref{lem:implications_deviations} confirms that the ask undercut inducing midpoint \(y\) at \(m_1\) is also strictly profitable.
\end{proof}

\nashExistence*
\begin{proof}
    Since \(V\) is continuous, \(a(\cdot)\), \(b(\cdot)\), and hence \(H\) are continuous, so \(\mathcal{F}\) is closed and thus compact: it has minimum and maximum elements \(m_{\min}\) and \(m_{\max}\), and contains the infimum and supremum of any of its nonempty subsets. \Cref{lem:extremal_immunity} gives \(m_{\max}\in \mathcal{F}_B\), so \(\mathcal{F}_B\neq\varnothing\); since \(\mathcal{F}_B\subseteq\mathcal{F}\) is nonempty, \(m^-= \inf \mathcal{F}_B\) exists and lies in \(\mathcal{F}\). We show \(m^-\in \mathcal{F}_B\).
    
    Suppose not: a profitable bid raise at \(m^-\) induces some mid-price \(x>m^-\). Since \(m^-\) is a lower bound of \(\mathcal{F}_B\) but \(x\) is not (else \(x\) would be a lower bound exceeding the infimum), there is \(m_B\in \mathcal{F}_B\) with \(m^-\le m_B<x\). If \(m_B=m^-\), this contradicts \(m^-\notin \mathcal{F}_B\) outright; if \(m^-<m_B<x\), \Cref{lem:deviations_contagion} makes the same bid raise profitable at \(m_B\), contradicting \(m_B\in \mathcal{F}_B\). Hence \(m^-\in \mathcal{F}_B\), so \(m^-=\min \mathcal{F}_B\) and, in particular, admits no profitable bid raise.

    To show that \(m^-\) also admits no profitable undercut, suppose, for contradiction, that a profitable undercut at \(m^-\) induces \(y<m^-\); by \Cref{lem:implications_deviations},
    \[
        2y-a(y)>b(m^-) \qquad\text{and}\qquad H(y)<y .
    \]
    Since \(H(y)<y\), \(y\notin\mathcal{F}\), so \(y\neq m_{\min}\). Moreover, \(m^-\neq m_{\min}\): otherwise the assumed undercut at \(m^-\) would itself be a profitable undercut at \(m_{\min}\), contradicting \Cref{lem:extremal_immunity}. Since \(\mathcal{F}_B\subseteq\mathcal{F}\), \(m^-=\inf \mathcal{F}_B\ge\inf\mathcal{F}=m_{\min}\), so together with \(m^-\neq m_{\min}\) we get \(m_{\min}<m^-\). Neither can \(y<m_{\min}\): \Cref{lem:deviations_contagion} (with \(m_1=m_{\min}<m^-=m_2\)) would then make the same undercut profitable at \(m_{\min}\), contradicting \Cref{lem:extremal_immunity}. Hence
    \[
        m_{\min}<y<m^- .
    \]

    Let \(\mathcal{F}_y\coloneq\{m\in\mathcal{F} \mid m<y\}\); it is nonempty (it contains \(m_{\min}\)), and its supremum lies in \(\mathcal{F}\). As \(y\notin\mathcal{F}\), this supremum cannot equal \(y\), so it is attained: \(\mathcal{F}_y\) has a maximum \(m_k\coloneq\max\mathcal{F}_y<y<m^-\). Since \(m_k<m^-=\min \mathcal{F}_B\), \(m_k\notin \mathcal{F}_B\): a profitable bid raise at \(m_k\) induces some \(x>m_k\), and by \Cref{lem:implications_deviations},
    \[
        2x-b(x)<a(m_k)\qquad\text{and}\qquad H(x)>x .
    \]

    Since \(H(x) > x\), \(x\) is neither \(y\), as \(H(y) < y\), nor \(m^-\), as \(m^- \in \mathcal{F}\). Thus, the remaining viable positions for \(x\) are \(x < y\), \(y < x < m^-\), or \(x > m^-\):
    \begin{enumerate}
        \item Case \(x < y\): Since \(H(x)>x\), \(H(y)<y\), and \(H\) is nondecreasing, \(H\) maps \([x,y]\) into itself; by Knaster--Tarski, \(H\) admits a fixed point \(\check{m}\in[x,y]\). Since \(\check{m}\ge x>m_k\) and \(\check{m}\le y\) with \(\check{m}\ne y\) (as \(y\notin\mathcal{F}\)), we have \(m_k<\check{m}<y\), so \(\check{m}\in\mathcal{F}_y\), contradicting the maximality of \(m_k\) there.
        
        \item Case \(y < x < m^-\): As \(m_k < y\) and \(x < m^-\), monotonicity of \(a(\cdot)\) and \(b(\cdot)\) implies \(2x - b(x) < a(m_k) \le a(y)\) and \(b(x) \le b(m^-) < 2y - a(y)\). Combining these gives \(2x - b(x) < a(y) < 2y - b(x)\), which reduces to \(x < y\), contradicting \(y < x\).
        
        \item Case \(x > m^-\): Since \(m_k < m^- < x\), \Cref{lem:deviations_contagion} implies that the bid raise inducing midpoint \(x\) is also profitable at \(m^-\), contradicting \(m^- \in \mathcal{F}_B\).
    \end{enumerate}
    Every case is contradictory, so \(m^-\) admits no profitable undercut. By \Cref{lem:single_sided_deviations}, \(m^-\) is therefore immune to every profitable unilateral deviation, so \(\mathbf{s}(m^-)\) is a pure-strategy Nash equilibrium.
\end{proof}

\nashExistenceDuality*
\begin{proof}
    By \Cref{lem:extremal_immunity}, \(m_{\min}\in \mathcal{F}_A\), so \(\mathcal{F}_A\neq\varnothing\); since \(\mathcal{F}\) is compact (as noted in the proof of \Cref{thm:nash_existence}) and \(\mathcal{F}_A\subseteq\mathcal{F}\) is nonempty, \(m^+\coloneq\sup \mathcal{F}_A\) exists and lies in \(\mathcal{F}\). We show \(m^+\in \mathcal{F}_A\).

    Suppose not: a profitable undercut at \(m^+\) induces some \(y<m^+\). Since \(m^+\) is an upper bound of \(\mathcal{F}_A\) but \(y\) is not (else \(y\) would be an upper bound below the supremum), there is \(m_A\in \mathcal{F}_A\) with \(y<m_A\le m^+\). If \(m_A=m^+\), this contradicts \(m^+\notin \mathcal{F}_A\) outright; if \(y<m_A<m^+\), \Cref{lem:deviations_contagion} makes the same undercut profitable at \(m_A\), contradicting \(m_A\in \mathcal{F}_A\). Hence \(m^+\in \mathcal{F}_A\), so \(m^+=\max \mathcal{F}_A\) and, in particular, admits no profitable undercut. It remains to exclude a profitable bid raise.

    Suppose one exists at \(m^+\), inducing \(x>m^+\). If \(m^+=m_{\max}\), this bid raise is itself profitable at \(m_{\max}\), contradicting \Cref{lem:extremal_immunity}; hence \(m^+<m_{\max}\). Moreover \(x\le m_{\max}\): otherwise, since \(m^+<m_{\max}<x\), \Cref{lem:deviations_contagion} would carry the bid raise to \(m_{\max}\), again contradicting \Cref{lem:extremal_immunity}. Thus \(\mathcal{F}_x \coloneq \{m \in \mathcal{F} \mid m \ge x\}\) contains \(m_{\max}\) and, being closed and bounded, is compact, so
    \[
        m_k \coloneq \min \mathcal{F}_x
    \]
    is well defined. Since \Cref{lem:implications_deviations} gives \(H(x)>x\), \(x\) is not a fixed point, so \(m_k\neq x\), so \(m_k>x>m^+=\max \mathcal{F}_A\); in particular, \(m_k\notin \mathcal{F}_A\), so it admits a profitable undercut inducing some \(y<m_k\).

    This is the mirror image of the construction in \Cref{thm:nash_existence}: there, an undercut at \(m^-\) produced \(y\) and a bid raise at an auxiliary \(m_k\) produced \(x\); here, a bid raise at \(m^+\) produces \(x\) and an undercut at the auxiliary \(m_k\) produces \(y\). Since the undercut at \(m_k\) gives \(H(y) < y\), \(y\) is neither \(x\), as \(H(x) > x\), nor \(m^+\), as \(m^+ \in \mathcal{F}\), and the same three-way split on the position of \(y\) closes off identically:
    \begin{enumerate}
        \item Case \(y<m^+\): \Cref{lem:deviations_contagion} carries the undercut at \(m_k\) to \(m^+\), contradicting \(m^+\in \mathcal{F}_A\).
        
        \item Case \(m^+<y<x\): Monotonicity of \(a(\cdot), b(\cdot)\) collapses the two deviation inequalities to \(x<y\), contradicting \(y<x\).
        
        \item Case \(x<y\): \(H\) maps \([x,y]\) into itself, so Knaster--Tarski gives a fixed point strictly between \(x\) and \(y\), contradicting the minimality of \(m_k\).
    \end{enumerate}
    Hence no profitable bid raise exists at \(m^+\), and \Cref{lem:single_sided_deviations} makes \(\mathbf{s}(m^+)\) a pure-strategy Nash equilibrium.
\end{proof}


\subsection{Proofs of Section~\ref{sec:learning}} \label{app:proofs_learning}

\nashPriceInterval*
\begin{proof}
    We first show (i). By \Cref{thm:existence_zero_profit_midpoint}, \(\mathcal{F}\) is nonempty; equivalently, \(G(m) = m - H(m)\) has a zero. The function \(G\) is continuous and, by the second condition of \Cref{ass:learning}, nondecreasing, so its zero set is a closed interval: if \(G(m_1) = G(m_2) = 0\) for some \(m_1 < m_2\), monotonicity forces \(0 = G(m_1) \le G(m) \le G(m_2) = 0\) for every \(m \in [m_1, m_2]\), so the zero set is convex, and continuity makes it closed. Hence \(\mathcal{F}\) is a nonempty closed interval \([m^-, m^+]\).

    We now show (ii). Fix \(m^* \in \mathcal{F}\), and consider a unilateral deviation \((\widetilde{b}, \widetilde{a})\) from the symmetric zero-profit profile \(\mathbf{s}(m^*)\), where both makers play \((b(m^*), a(m^*))\):
    \begin{enumerate}
        \item Case \(\widetilde{a} \ge a(m^*)\) and \(\widetilde{b} \le b(m^*)\): The market prices, and hence the mid-price \(m^*\), are unchanged, so the deviator earns zero profit on each side, and the deviation is not profitable.

        \item Case \(\widetilde{a} < a(m^*)\) and \(\widetilde{b} \le b(m^*)\): On the ask side, the deviator is now uniquely best, so the market mid-price becomes
        \[
            y \coloneq \dfrac{\widetilde{a} + b(m^*)}{2} < m^*,
        \]
        and its ask-side profit is \(\mathbb{P}(V > y) \left( \widetilde{a} - a(y) \right)\); we show this margin is negative. Since \(G\) is nondecreasing by the second condition of \Cref{ass:learning} and \(y < m^*\), \(G(y) \le G(m^*) = 0\), i.e.,
        \[
            y \le H(y) = \dfrac{a(y) + b(y)}{2}
            \iff
            a(y) \ge 2y - b(y).
        \]
        Since \(b\) is strictly increasing by the first condition of \Cref{ass:learning} and \(y < m^*\), \(b(y) < b(m^*)\). Hence
        \[
            \widetilde{a} = 2y - b(m^*) < 2y - b(y) \le a(y),
        \]
        so \(\widetilde{a} - a(y) < 0\), and the ask side is loss-making. On the bid side, the deviator is either inactive, or ties at \(b(m^*)\) and earns \(\frac{1}{2} \, \mathbb{P}(V < y) \left( b(y) - b(m^*) \right) < 0\). Hence the deviation is not profitable.

        \item Case \(\widetilde{a} \ge a(m^*)\) and \(\widetilde{b} > b(m^*)\): Symmetric to the previous case, with the roles of \(a\) and \(b\) swapped.

        \item Case \(\widetilde{a} < a(m^*)\) and \(\widetilde{b} > b(m^*)\): The deviator is now uniquely best on both sides, so the market mid-price becomes
        \[
            z \coloneq \dfrac{\widetilde{a} + \widetilde{b}}{2},
        \]
        and its total profit is
        \[
            \mathbb{P}(V > z) \left( \widetilde{a} - a(z) \right) +
            \mathbb{P}(V < z) \left( b(z) - \widetilde{b} \right).
        \]
        We show that both margins are negative when \(z \ge m^*\); the case \(z < m^*\) is symmetric. Since \(a(\cdot)\) is strictly increasing by the first condition of \Cref{ass:learning} and \(z \ge m^*\), \(a(z) \ge a(m^*) > \widetilde{a}\), so the ask side is loss-making. For the bid side, since \(\widetilde{b} = 2z - \widetilde{a}\) and \(b(m^*) = 2m^* - a(m^*)\),
        \[
            \widetilde{b} - b(m^*)
            = 2(z - m^*) +
            \underbrace{\left( a(m^*) - \widetilde{a} \right)}_{> 0}
            > 2(z - m^*).
        \]
        By \Cref{ass:learning}, \(b'(m) \le 2 - a'(m) < 2\) on \((0, 1)\), so the mean value theorem gives \(b(z) - b(m^*) \le 2(z - m^*)\). Hence
        \[
            \widetilde{b} - b(m^*) > 2(z - m^*) \ge b(z) - b(m^*),
        \]
        so \(\widetilde{b} > b(z)\): the bid side is also loss-making, and the deviation is not profitable.
    \end{enumerate}
    Hence \(\mathbf{s}(m^*)\) is a PSNE, and \(m^*\) is a PSNE mid-price.
\end{proof}

\noCycle*
\begin{proof}
    Fix \(A \in (0, 1]\), and let \(m(A) \coloneq (A + \beta(A))/2\) be the mid-price induced by \(A\) and \(\beta(A)\), so that \(\beta(A) = b(m(A))\) by \Cref{def:zero_profit_responses}.

    We first reduce the claim to the sign of \(G\). As shown after \Cref{def:zero_profit_responses}, the map \(x \mapsto x - a((x + \beta(A))/2)\) is strictly increasing and vanishes exactly at \(x = \alpha(\beta(A))\). At \(x = A\), it equals
    \[
        A - a(m(A)) = 2 m(A) - b(m(A)) - a(m(A)) = 2 G(m(A)).
    \]
    Hence \(\alpha(\beta(A)) > A\), \(= A\), or \(< A\) according to whether \(G(m(A))\) is negative, zero, or positive.

    We now locate \(m(A)\). Since \(\beta\) is strictly increasing, so is \(m(\cdot)\). Moreover, \(\beta(a^\pm) = b^\pm\), so \(m(a^\pm) = H(m^\pm) = m^\pm\). Hence \(m(A) < m^-\) if \(A < a^-\), \(m(A) \in [m^-, m^+]\) if \(A \in [a^-, a^+]\), and \(m(A) > m^+\) if \(A > a^+\).

    Finally, \(G\) is nondecreasing by the second condition of \Cref{ass:learning}, and its zero set is exactly \(\mathcal{F} = [m^-, m^+]\) by \Cref{prop:nash_structure_regularity}. Hence \(G < 0\) on \([0, m^-)\), \(G = 0\) on \([m^-, m^+]\), and \(G > 0\) on \((m^+, 1]\), which gives the three cases of the claim.
\end{proof}

\aggressiveTick*
\begin{proof}
    We prove the ask statement; the bid statement is symmetric. Only \(a_i\) changes, so the market bid \(B\) and maker \(i\)'s share of the bid side, \(\lambda \coloneq \mathbb{I}_{\{b_i = B\}} / n_B(\mathbf{b}) \in \{0, \frac{1}{2}, 1\}\), remain fixed throughout. For \(\widetilde{a} \in [a_i, a_i + \Delta]\), let \(y \coloneq (\widetilde{a} + B)/2\), and let
    \[
        u_i(\widetilde{a})
        \coloneq
        u_i^{\mathrm{ask}}(\widetilde{a}) + u_i^{\mathrm{bid}}(\widetilde{a})
        =
        (1 - F(y))(\widetilde{a} - a(y)) + \lambda \, F(y)(b(y) - B)
    \]
    be maker \(i\)'s payoff if it alone quoted the market ask \(\widetilde{a}\), all other quotes unchanged, split into its ask-side and bid-side components.

    We first show that the ask side is loss-making, i.e., \(u_i^{\mathrm{ask}}(\widetilde{a}) < 0\), or equivalently \(\widetilde{a} < a(y)\). If \(B < 1\), the map \(\varphi(x) \coloneq x - a((x + B)/2)\) is strictly increasing on \([0, 1]\) and vanishes at \(x = \alpha(B)\), as noted after \Cref{def:zero_profit_responses}. Since \(\widetilde{a} \le a_i + \Delta < \alpha(B)\),
    \[
        \widetilde{a} - a(y) = \varphi(\widetilde{a}) < \varphi(\alpha(B)) = 0.
    \]
    If \(B = 1\), then \(\widetilde{a} < \alpha(1) = 1\), so \(\widetilde{a} < y < 1\) and hence \(\widetilde{a} < y < a(y)\).

    We now show that \(u_i\) is strictly increasing on \([a_i, a_i + \Delta]\). For \(\widetilde{a} \in (a_i, a_i + \Delta)\), we have \(0 < \widetilde{a} < \alpha(B) \le 1\) and \(0 \le B \le 1\), so \(y \in (0, 1)\). Writing \((1 - F(y))\, a(y) = \int_y^1 v f(v)\,dv\) and \(F(y)\, b(y) = \int_0^y v f(v)\,dv\), and using \(dy/d\widetilde{a} = 1/2\) and \(\widetilde{a} - y = y - B = (\widetilde{a} - B)/2\), we obtain
    \begin{align*}
        u_i'(\widetilde{a})
        &= 1 - F(y) - \dfrac{1}{2} f(y)(\widetilde{a} - y) + \dfrac{\lambda}{2} f(y)(y - B)\\
        &= 1 - F(y) - \dfrac{1 - \lambda}{4} \, f(y)(\widetilde{a} - B).
    \end{align*}
    The first term is the margin gained on every executed buy order, while the second is the net loss from the marginal trader at \(y\), who switches from buying at \(\widetilde{a}\) to selling at \(B\), of which maker \(i\) captures only a share \(\lambda\). If \(\widetilde{a} \le B\), then \(u_i'(\widetilde{a}) \ge 1 - F(y) > 0\). Otherwise, the ask-side loss shown above gives \(\widetilde{a} - y < a(y) - y\), so, using \(a'(y) = f(y)(a(y) - y)/(1 - F(y))\),
    \begin{align*}
        \dfrac{1 - \lambda}{4} \, f(y)(\widetilde{a} - B)
        &\le \dfrac{1}{2} f(y)(\widetilde{a} - y)\\
        &< \dfrac{1}{2} f(y)(a(y) - y)
        = \dfrac{1}{2} (1 - F(y)) \, a'(y).
    \end{align*}
    Hence \(u_i'(\widetilde{a}) > (1 - F(y))(1 - a'(y)/2) > 0\), since \(a'(y) \le 2 - b'(y) < 2\) by \Cref{ass:learning}. As \(u_i\) is continuous on \([a_i, a_i + \Delta]\), the mean value theorem shows that it is strictly increasing there.

    Finally, since \(a_i = A\) and prices lie on the grid, exactly one of the following holds:
    \begin{enumerate}
        \item Case \(a_{-i} \ge a_i + 2\Delta\): Maker \(i\) alone quotes the market ask before and after the deviation, so its payoff rises from \(u_i(a_i)\) to \(u_i(a_i + \Delta)\).

        \item Case \(a_{-i} = a_i + \Delta\): After the deviation, maker \(i\) ties at the market ask \(a_i + \Delta\). The mid-price is the same as if it were alone there, so its bid-side payoff is \(u_i^{\mathrm{bid}}(a_i + \Delta)\), while its ask-side payoff is halved. Its payoff is therefore \(u_i(a_i + \Delta) - \frac{1}{2} u_i^{\mathrm{ask}}(a_i + \Delta) > u_i(a_i + \Delta) > u_i(a_i)\).

        \item Case \(a_{-i} = a_i\): Before the deviation, maker \(i\) ties at \(a_i\) and earns \(\frac{1}{2} u_i^{\mathrm{ask}}(a_i) + u_i^{\mathrm{bid}}(a_i)\). After it, the opponent alone quotes the market ask \(a_i\), so the mid-price, and hence maker \(i\)'s bid-side payoff, is unchanged, while its ask-side payoff becomes zero, for a gain of \(-\frac{1}{2} u_i^{\mathrm{ask}}(a_i) > 0\).
    \end{enumerate}
    In every case, the deviation strictly increases \(u_i\), and it is feasible, since \(b_i \le a_i < a_i + \Delta < \alpha(B) \le 1\).
\end{proof}

\passiveTick*
\begin{proof}
    We prove the ask statement; the bid statement is symmetric. The deviation changes only \(a_i\), so the market bid \(B\) and maker \(i\)'s share \(\lambda\) of the bid side, as in the proof of \Cref{lem:aggressive_tick}, are fixed throughout. As in that proof, if maker \(i\) alone quoted the market ask \(\widetilde{a}\), its ask-side and bid-side payoffs would be
    \[
        u_i^{\mathrm{ask}}(\widetilde{a}) = (1 - F(y))(\widetilde{a} - a(y)),
    \quad \text{and} \quad
        u_i^{\mathrm{bid}}(\widetilde{a}) = \lambda \, F(y)(b(y) - B),
    \]
    where \(y \coloneq (\widetilde{a} + B)/2\).

    We first fix the constants. For \(\widetilde{a} \in [\alpha(B), 1]\), since \(\alpha(0) \le \alpha(B)\) and \(B \le \overline{B}\), the mid-price \(y\) lies in the compact interval
    \[
        K \coloneq \left[ \alpha(0)/2, \, (1 + \overline{B})/2 \right] \subset (0, 1).
    \]
    On \(K\), the functions \(1 - F\) and \(1 - a'/2\) are continuous and positive, the latter since \(a'(m) \le 2 - b'(m) < 2\) by \Cref{ass:learning}, so they are bounded below by some \(c_1, c_2 > 0\), respectively; moreover, \(f\) is bounded. Since the margin \(\widetilde{a} - a(y)\) vanishes at \(\widetilde{a} = \alpha(B)\) and has derivative \(1 - a'(y)/2 \ge c_2\), every \(\widetilde{a} \in [\alpha(B), 1]\) satisfies
    \[
        u_i^{\mathrm{ask}}(\widetilde{a}) \ge c \, (\widetilde{a} - \alpha(B)),
        \qquad c \coloneq c_1 c_2.
    \]
    Moreover, differentiating as in the proof of \Cref{lem:aggressive_tick},
    \[
        \frac{d}{d\widetilde{a}} u_i^{\mathrm{ask}}(\widetilde{a})
        = 1 - F(y) - \frac{1}{4} f(y)(\widetilde{a} - B),
    \quad \text{and} \quad
        \frac{d}{d\widetilde{a}} u_i^{\mathrm{bid}}(\widetilde{a})
        = \frac{\lambda}{4} f(y)(\widetilde{a} - B),
    \]
    which are both bounded in absolute value by some \(L > 0\) on \([\alpha(B), 1]\). The constants \(c\) and \(L\) depend only on \(F\) and \(\overline{B}\), and we choose \(C_0 > 1\) such that \(c C_0 > 4L\); for the bid statement, the symmetric construction yields constants depending only on \(F\) and \(\underline{A}\), and we take \(C_0\) large enough for both.

    Consider now the deviation from \(a_i\) to \(\widetilde{a} - \Delta\). Since \(C_0 > 1\),
    \[
        \widetilde{a} - \Delta \ge \alpha(B) + (C_0 - 1) \Delta > \alpha(B) > B \ge b_i,
    \]
    so the deviation is feasible, and \(u_i^{\mathrm{ask}}(\widetilde{a} - \Delta) > 0\) by the first bound above. Since prices lie on the grid and \(a_{-i} \le \widetilde{a} \le a_i\), one of the following cases holds:
    \begin{enumerate}
        \item Case \(a_{-i} \le \widetilde{a} - 2\Delta\): Maker \(i\) quotes above the market ask before and after the deviation, and market prices are unchanged, so its payoff is unchanged.

        \item Case \(a_{-i} = \widetilde{a} - \Delta\): Before the deviation, maker \(i\) quotes above the market ask, so its ask-side payoff is zero. After the deviation, it ties at the market ask, which is unchanged, so its bid-side payoff is unchanged and its ask-side payoff becomes \(\frac{1}{2} u_i^{\mathrm{ask}}(\widetilde{a} - \Delta) > 0\).

        \item Case \(a_{-i} = \widetilde{a}\): Before the deviation, maker \(i\)'s ask-side payoff is \(\frac{1}{2} u_i^{\mathrm{ask}}(\widetilde{a})\) if \(a_i = \widetilde{a}\) and zero otherwise, hence at most \(\frac{1}{2} u_i^{\mathrm{ask}}(\widetilde{a})\), since \(u_i^{\mathrm{ask}}(\widetilde{a}) \ge c C_0 \Delta > 0\). After the deviation, it alone quotes the market ask \(\widetilde{a} - \Delta\); by the mean value theorem, its ask-side payoff is \(u_i^{\mathrm{ask}}(\widetilde{a} - \Delta) \ge u_i^{\mathrm{ask}}(\widetilde{a}) - L\Delta\), while its bid-side payoff decreases by at most \(L \Delta\). Hence its total gain is at least
        \[
            \frac{1}{2} u_i^{\mathrm{ask}}(\widetilde{a}) - 2 L \Delta
            \ge \left( \frac{c C_0}{2} - 2L \right) \Delta > 0.
        \]
    \end{enumerate}
    Hence, the deviation weakly increases \(u_i\), strictly if \(a_{-i} \ge \widetilde{a} - \Delta\).
\end{proof}

\cePriceBounds*
\begin{proof}
    We prove the ask bounds; the bid bounds are symmetric. We drop the argument \(\pi\) of the extremal prices, and call a profile \emph{supported} if it lies in \(\supp(\pi)\). We repeatedly use the following consequence of \Cref{def:ce}: if, for a maker \(i\), a supported profile \(\mathbf{s}\), and some \(\widetilde{s}_i \in \mathcal{S}_i^\Delta\), switching from \(s_i\) to \(\widetilde{s}_i\) weakly increases \(u_i\) at every supported profile in which maker \(i\) plays \(s_i\), and strictly at \(\mathbf{s}\), then the CE constraint for \(s_i\) and \(\widetilde{s}_i\) is violated. Let \(\Delta_0 \coloneq \min\{\alpha(0), 1 - \beta(1)\}/2\), and let \(C_0\) be the constant of \Cref{lem:passive_tick} for \(\overline{B} \coloneq (1 + \beta(1))/2\) and \(\underline{A} \coloneq \alpha(0)/2\); both depend only on \(F\).

    We first show that \(a^{\min} \ge \alpha(b^{\min}) - \Delta\). Suppose not, and let \(\mathbf{s}\) be a supported profile with market ask \(a^{\min}\), quoted by maker \(i\). Every supported profile in which maker \(i\) plays \(s_i\) has market ask at least \(a^{\min} = a_i\), so maker \(i\) quotes the market ask there, and its market bid \(B \ge b^{\min}\) satisfies
    \[
        a_i + \Delta < \alpha(b^{\min}) \le \alpha(B).
    \]
    By \Cref{lem:aggressive_tick}, raising \(a_i\) to \(a_i + \Delta\) strictly increases \(u_i\) at every such profile, a contradiction. Symmetrically, \(b^{\max} \le \beta(a^{\max}) + \Delta\).

    By these two bounds and \(\Delta \le \Delta_0\), every supported profile satisfies
    \[
        B \le b^{\max} \le \beta(1) + \Delta \le \overline{B}
        \qquad \text{and} \qquad
        A \ge a^{\min} \ge \alpha(0) - \Delta \ge \underline{A},
    \]
    so \Cref{lem:passive_tick} applies at every supported profile.

    We now show that \(a^{\max} \le \alpha(b^{\max}) + C_0 \Delta\). Suppose not; since \(\alpha\) is increasing, every supported profile then satisfies \(a^{\max} > \alpha(B) + C_0 \Delta\). It suffices to find a supported profile \(\mathbf{s}\) with market ask \(a^{\max}\) and a maker \(i\) such that, at every supported profile in which maker \(i\) plays \(s_i\), the opponent's ask is at most \(a^{\max}\). Indeed, since \(a_i \ge a^{\max}\), \Cref{lem:passive_tick} with \(\widetilde{a} \coloneq a^{\max}\) then shows that lowering \(a_i\) to \(a^{\max} - \Delta\) weakly increases \(u_i\) at every such profile, and strictly at \(\mathbf{s}\), where \(a_{-i} \ge a^{\max}\), a contradiction. We find such \(\mathbf{s}\) and \(i\) as follows:
    \begin{enumerate}
        \item Case some supported profile \(\mathbf{s}\) with market ask \(a^{\max}\) has a maker \(i\) with \(a_i > a^{\max}\): Every supported profile in which maker \(i\) plays \(s_i\) has market ask at most \(a^{\max} < a_i\), by definition of \(a^{\max}\); hence the market ask there is quoted by the opponent, whose ask is therefore at most \(a^{\max}\).

        \item Case both makers quote \(a^{\max}\) at every supported profile with market ask \(a^{\max}\): Take any such \(\mathbf{s}\) and any maker \(i\). If, at some supported profile in which maker \(i\) plays \(s_i\), the opponent's ask exceeded \(a^{\max}\), that profile would have market ask \(a_i = a^{\max}\), quoted by maker \(i\) alone, contradicting the case assumption.
    \end{enumerate}
    Hence \(a^{\max} \le \alpha(b^{\max}) + C_0 \Delta\).
\end{proof}

\gridCE*
\begin{proof}
    It suffices to show that, for any tick sizes \(\Delta_n \to 0\) with \(1/\Delta_n \in \mathbb{N}\) and any CEs \(\pi_n\) of \(\mathcal{G}_{\Delta_n}\),
    \[
        \max_{\mathbf{s} \in \supp(\pi_n)} D_{\mathrm{env}}(\mathbf{s}) \to 0;
    \]
    choosing each \(\pi_n\) within \(\Delta_n\) of the supremum in \(\varepsilon_{\Delta_n}\) then gives \(\varepsilon_{\Delta_n} \to 0\). Write \(a_n^{\min} \coloneq a^{\min}(\pi_n)\), and similarly \(a_n^{\max}\), \(b_n^{\min}\) and \(b_n^{\max}\). Let \(C_0\) and \(\Delta_0\) be the constants of \Cref{lem:ce_price_bounds}; discarding finitely many terms, we may assume that \(\Delta_n \le \Delta_0\) and
    \[
        (1 + C_0) \Delta_n \le \frac{1 - \beta(1)}{2}.
    \]

    We first combine the four bounds of \Cref{lem:ce_price_bounds}, namely
    \begin{align*}
        \beta(a_n^{\min}) &\le b_n^{\min} + C_0 \Delta_n, &
        \alpha(b_n^{\min}) &\le a_n^{\min} + \Delta_n, \\
        b_n^{\max} &\le \beta(a_n^{\max}) + \Delta_n, &
        a_n^{\max} &\le \alpha(b_n^{\max}) + C_0 \Delta_n.
    \end{align*}
    In particular, \(b_n^{\min} \le b_n^{\max} \le \beta(1) + \Delta_n\), so, by the choice of \(\Delta_n\), the points \(b_n^{\min} + C_0 \Delta_n\) and \(\beta(a_n^{\max}) + \Delta_n\), and hence also \(b_n^{\min}\) and \(\beta(a_n^{\max})\), lie in \(J \coloneq [0, (1 + \beta(1))/2] \subset [0, 1)\), on which \(\alpha\) is Lipschitz with some constant \(L_\alpha\). Since \(\alpha\) is also increasing, combining the bounds in each row gives, for the first row,
    \begin{align*}
        \alpha(\beta(a_n^{\min}))
        &\le \alpha(b_n^{\min} + C_0 \Delta_n)\\
        &\le \alpha(b_n^{\min}) + L_\alpha C_0 \Delta_n
        \le a_n^{\min} + (1 + L_\alpha C_0) \Delta_n,
    \end{align*}
    and similarly, for the second row,
    \[
        a_n^{\max} \le \alpha(\beta(a_n^{\max})) + (L_\alpha + C_0) \Delta_n.
    \]
    In words, the extremal asks are approximate fixed points of \(\alpha \circ \beta\), up to an error of order \(\Delta_n\).

    We now show that \(\liminf_n a_n^{\min} \ge a^-\) and \(\limsup_n a_n^{\max} \le a^+\). If \(\liminf_n a_n^{\min} < a^-\), then \(a_n^{\min} \to \bar{a} < a^-\) along a subsequence, with \(\bar{a} \ge \alpha(0) > 0\), since \(a_n^{\min} \ge \alpha(b_n^{\min}) - \Delta_n \ge \alpha(0) - \Delta_n\) by \Cref{lem:ce_price_bounds} and the monotonicity of \(\alpha\). Since \(\alpha \circ \beta\) is continuous on \((0, 1]\), the first inequality gives \(\alpha(\beta(\bar{a})) \le \bar{a}\) in the limit, contradicting \Cref{lem:no_cycle}. Symmetrically, if \(a_n^{\max} \to \bar{a} > a^+\) along a subsequence, with \(\bar{a} \le 1\), the second inequality gives \(\bar{a} \le \alpha(\beta(\bar{a}))\), again contradicting \Cref{lem:no_cycle}.

    The bid bounds follow from the ask bounds. By \Cref{lem:ce_price_bounds}, \(b_n^{\min} \ge \beta(a_n^{\min}) - C_0 \Delta_n\) and \(b_n^{\max} \le \beta(a_n^{\max}) + \Delta_n\). Since \(\beta\) is increasing and continuous, and \(\beta(a^\pm) = b^\pm\),
    \[
        \liminf_n b_n^{\min} \ge b^-,
        \qquad
        \limsup_n b_n^{\max} \le b^+.
    \]

    Finally, every \(\mathbf{s} \in \supp(\pi_n)\) has \(a_n^{\min} \le a^*(\mathbf{a}) \le a_n^{\max}\) and \(b_n^{\min} \le b^*(\mathbf{b}) \le b_n^{\max}\). Since \(\mathcal{R}^{\mathrm{NE}}\) is a rectangle, \(D_{\mathrm{env}}(\mathbf{s})\) is the sum of the distances of \(b^*(\mathbf{b})\) from \([b^-, b^+]\) and of \(a^*(\mathbf{a})\) from \([a^-, a^+]\), so
    \begin{align*}
        \max_{\mathbf{s} \in \supp(\pi_n)} D_{\mathrm{env}}(\mathbf{s})
        &\le \max\{a^- - a_n^{\min}, 0\} + \max\{a_n^{\max} - a^+, 0\} \\
        &\quad + \max\{b^- - b_n^{\min}, 0\} + \max\{b_n^{\max} - b^+, 0\},
    \end{align*}
    and the right-hand side tends to zero by the four limits above.
\end{proof}

\approxCE*
\begin{proof}
    Pass to a subsequence along which \(\mathbb{E}_{\mathbf{s} \sim \pi_n}[D_{\mathrm{env}}(\mathbf{s})]\) converges to its limit superior. Since \(\mathcal{S}^\Delta\) is finite, the set of distributions over it is compact, so, passing to a further subsequence, \(\pi_n \to \pi\). Both the inequalities of \Cref{def:ce} and the expectation of \(D_{\mathrm{env}}\) are linear in the distribution, so, since \(\delta_n \to 0\), \(\pi\) is a CE, and
    \[
        \limsup_n \mathbb{E}_{\mathbf{s} \sim \pi_n}[D_{\mathrm{env}}(\mathbf{s})]
        = \mathbb{E}_{\mathbf{s} \sim \pi}[D_{\mathrm{env}}(\mathbf{s})] 
        \le \max_{\mathbf{s} \in \supp(\pi)} D_{\mathrm{env}}(\mathbf{s}) 
        \le \varepsilon_\Delta,
    \]
    where the last inequality holds by definition of \(\varepsilon_\Delta\).
\end{proof}


\end{document}